\documentclass[12pt]{amsart}
\usepackage{amssymb,latexsym,amsmath,amscd,amsthm,graphicx, color}
\usepackage[all]{xy}
\usepackage{pgf,tikz}
\usepackage{mathrsfs}
\usetikzlibrary{arrows}
\definecolor{uuuuuu}{rgb}{0.26666666666666666,0.26666666666666666,0.26666666666666666}
\definecolor{xdxdff}{rgb}{0.49019607843137253,0.49019607843137253,1.}
\definecolor{ffqqqq}{rgb}{1.,0.,0.}
\usepackage{float}

\usepackage[linkcolor=blue, urlcolor=blue, citecolor=blue, colorlinks, bookmarks]{hyperref}

\definecolor{uuuuuu}{rgb}{0.26666666666666666,0.26666666666666666,0.26666666666666666}
\definecolor{qqwuqq}{rgb}{0.,0.39215686274509803,0.}
\definecolor{zzttqq}{rgb}{0.6,0.2,0.}
\definecolor{xdxdff}{rgb}{0.49019607843137253,0.49019607843137253,1.}
\definecolor{qqqqff}{rgb}{0.,0.,1.}
\definecolor{cqcqcq}{rgb}{0.7529411764705882,0.7529411764705882,0.7529411764705882}
\definecolor{sqsqsq}{rgb}{0.12549019607843137,0.12549019607843137,0.12549019607843137}

\allowdisplaybreaks
\theoremstyle{plain}

\newtheorem{theorem}[subsection]{Theorem}

\newtheorem{lemma}[subsection]{Lemma}

\newtheorem{defi}[subsection]{Definition}

\newtheorem{prop}[subsection]{Proposition}

\theoremstyle{definition}
\newtheorem{remark}[subsection]{Remark}

\newtheorem{note}[subsection]{Note}

\newcommand{\uu}{\cup}
\newcommand{\ii}{\cap}
\newcommand{\UU}{\bigcup}
\newcommand{\II}{\bigcap}

\newcommand{\sci}{\subset}
\newcommand{\es}{\emptyset}
\newcommand{\set}[1]{\{#1\}}

\newcommand{\ga}{\alpha}
\newcommand{\gb}{\beta}

\newcommand{\gd}{\delta}
\renewcommand{\gg}{\gamma}

\newcommand{\go}{\omega}

\newcommand{\gt}{\tau}

\newcommand{\tit}{\textit}

\newcommand{\C}[1]{\mathcal{#1}}
\newcommand{\D}[1]{\mathbb{#1}}

\newcommand{\te}{\text}

\newcommand{\nd}{\noindent}

\newcommand{\tri}{\triangle}
\begin{document}
\nd To appear, Discrete and Continuous Dynamical Systems - Series S (DCDS-S)
\title{Optimal quantization for a probability measure on a nonuniform stretched Sierpi\'{n}ski triangle}

\author[Megha Pandey and Mrinal Kanti Roychowdhury]{Megha Pandey and Mrinal Kanti Roychowdhury}



\thanks{$^*$Corresponding author: Megha Pandey}
\date{}
\maketitle

\pagestyle{myheadings}\markboth{Megha Pandey and Mrinal Kanti Roychowdhury}{Optimal quantization for a probability measure on a nonuniform stretched Sierpi\'{n}ski triangle}

\begin{abstract}
Quantization for a Borel probability measure refers to the idea of estimating a given probability by a discrete probability with support containing a finite number of elements. In this paper, we have considered a Borel probability measure $P$ on $\mathbb R^2$, which has support a nonuniform stretched Sierpi\'{n}ski triangle generated by a set of three contractive similarity mappings on $\mathbb R^2$. For this probability measure, we investigate the optimal sets of $n$-means and the $n$th quantization errors for all positive integers $n$.
\end{abstract}

\section{Introduction}  
Optimal quantization is a fundamental problem in signal processing, data compression, and information theory. We refer to \cite{GG, GN, Z2} for surveys on the subject and comprehensive lists of references to the literature; see also \cite{AW, GKL, GL1, Z1}.
Recently, Pandey and Roychowdhury have introduced the concepts of constrained quantization and conditional quantization (see \cite{PR4, PR2, PR1}). A quantization without a constraint is known as an unconstrained quantization, i.e., unconstrained quantization, which traditionally in the literature is known as quantization, plays as a special case of constrained quantization. To know the details of constrained quantization on can see \cite{PR1}, and for unconstrained quantization one can see \cite{GL}.
After the introduction of constrained quantization and then conditional quantization, the quantization theory is now much more enriched with huge applications in our real world. For some follow up papers in the direction of constrained quantization and conditional quantization, one can see \cite{BCDRV, BCDR, HNPR, PR3, PR5}.  
On unconstrained quantization, there is a number of papers written by many authors; for example, one can see \cite{DR, DFG, GG, GL2, GL3, GL, GN, GL1, KNZ, P, P1, R1, R2, R3, RS, Z1, Z2}.
 \begin{defi}  \label{defi0} 
Let $P$ be a Borel probability measure on a $k$-dimensional Euclidean space $\D R^k$, where $k\in \D N$, equipped with a Euclidean metric $d$ induced by the Euclidean norm $\|\cdot\|$.  
Then, for $n\in \D N$, the \tit {$n$th quantization
	error} for $P$ is defined by
\begin{equation}  \label{eq0} 
V_{n}:=V_{n}(P)=\inf\Big\{\int \mathop{\min}\limits_{a\in\ga} \|x-a\|^2 dP(x) : \ga\sci \D R^k \te{ and } 1\leq \te{card}(\ga) \leq n \Big\},
\end{equation} 
where $\te{card}(A)$ represents the cardinality of the set $A$. 
\end{defi} 
We assume that $\int d(x, 0)^2 dP(x)<\infty$ to make sure that the infimum in \eqref{eq0} exists (see \cite{GL, PR1}). Such a set $\ga$ for which the infimum occurs and contains no more than $n$ elements is called an \tit{optimal set of $n$-means}.  The collection of all optimal sets of $n$-means for a Borel probability measure $P$ is denoted by $\C C_n:=\C C_n(P)$. The elements of an optimal set are called \tit{optimal elements}.  If $\ga$ is a finite set, in general, the error $\int \min_{a \in \ga} \|x-a\|^2 dP(x)$ is often referred to as the \tit{cost} or \tit{distortion error} for $\ga$, and is denoted by $V(P; \ga)$. Thus, $V_n:=V_n(P)=\inf\set{V(P; \ga) : \ga\sci \D  R^k, 1\leq \te{ card}(\alpha) \leq n}$. It is known that for a Borel probability measure with support containing at least n elements, an optimal set of n-means always has exactly $n$ elements (see \cite{GL, PR1}). The number
\[\lim_{n\to \infty} \frac{2\log n}{-\log V_n(P)}, \]
if it exists, is called the \tit{quantization dimension} of the probability measure $P$. The quantization dimension measures the speed how fast the specified measure of the error tends to zero as $n$ approaches infinity.  Given a finite subset $\ga\sci \D R^k$, the \tit{Voronoi region} generated by $a\in \ga$ is defined by
\[M(a|\ga)=\set{x \in \D R^k : \|x-a\|=\min_{b \in \ga}\|x-b\|}\]
i.e., the Voronoi region generated by $a\in \ga$ is the set of all elements $x$ in $\D R^k$ such that $a$ is a nearest element to $x$ in $\ga$, and the set $\set{M(a|\ga) : a \in \ga}$ is called the \tit{Voronoi diagram} or \tit{Voronoi tessellation} of $\D R^k$ with respect to $\ga$. A Voronoi tessellation is called a \tit{centroidal Voronoi tessellation} (CVT) if the generators of the tessellation are also the centroids of their own Voronoi regions with respect to the probability measure $P$. A Borel measurable partition $\set{A_a : a\in \ga}$, where $\ga$ is an index set, of $\D R^k$ is called a Voronoi partition of $\D R^k$ if $A_a\sci M(a|\ga)$ for every $a\in \ga$.
Let us now state the following proposition (see \cite{GG, GL}):
\begin{prop} \label{prop10}
Let $\alpha$ be an optimal set of $n$-means and $a\in \ga$. Then,

$(i)$ $P(M(a|\ga))>0$, $(ii)$ $ P(\partial M(a|\ga))=0$, $(iii)$ $a=E(X : X \in M(a|\ga))$, and $(iv)$ $P$-almost surely the set $\set{M(a|\ga) : a \in \ga}$ forms a Voronoi partition of $\D R^k$.
\end{prop}
Let $\alpha$ be an optimal set of $n$-means and  $a \in \alpha$, then by Proposition~\ref{prop10}, we see that $a$ is the conditional expectation in its own Voronoi region, i.e., 
\begin{align*}
a=\frac{1}{P(M(a|\ga))}\int_{M(a|\ga)} x dP=\frac{\int_{M(a|\ga)} x dP}{\int_{M(a|\ga)} dP},
\end{align*} 
which also implies that $a$ is the centroid of the Voronoi region $M(a|\ga)$ associated with the probability measure $P$ (see  \cite{DFG, R1}).

Let $P$ be a Borel probability measure on $\D R$ given by $P=\frac 12 P\circ S_1^{-1}+\frac 12 P\circ S_2^{-1}$, where $S_1(x)=\frac 13x$ and $S_2(x)=\frac 13 x +\frac 23$ for all $x\in \D R$. Then, $P$ has support the classical Cantor set $C$.  For this probability measure Graf and Luschgy gave an exact formula to determine the optimal sets of $n$-means and the $n$th quantization errors for all $n\geq 2$; they also proved that the quantization dimension of this distribution exists and is equal to the Hausdorff dimension $\gb:=\log 2/(\log 3)$ of the Cantor set, but the $\gb$-dimensional quantization coefficient does not exist (see \cite{GL2}). The bounds of the above exact formula are given in \cite{R2}. In \cite{LR} for $n\geq 2$, L. Roychowdhury gave an induction formula to determine the optimal sets of $n$-means and the $n$th quantization errors for a Borel probability measure $P$ on $\D R$, given by $P=\frac 14 P\circ S_1^{-1}+\frac 34 P\circ S_2^{-1}$ which has support the Cantor set generated by $S_1$ and $S_2$, where $S_1(x)=\frac 14 x$ and $S_2(x)=\frac 12 x+\frac 12$ for all $x \in \D R$. In \cite{R3}, M. Roychowdhury gave an infinite extension of the result of Graf-Luschgy in \cite{GL2}. In \cite{CR1}, \c C\"omez and Roychowdhury gave an exact formula to determine the optimal sets of $n$-means and the $n$th quantization error for a Borel probability measure supported by a Cantor dust.  In \cite{R4}, for a nonuniform  probability measure $P$ on $\mathbb R^2$ which has support a Sierpi\'nski carpet generated by a set of four contractive similarity mappings with equal similarity ratios, Roychowdhury investigated the optimal sets of $n$-means and the $n$th quantization errors for all $n\geq 2$.
\par
Let us now consider a set of three contractive similarity mappings $S_1, S_2, S_3$ on $\D R^2$, such that
 $S_1(x_1, x_2)=\frac 13 (x_1, x_2)$, $S_2(x_1, x_2)=\frac 13 (x_1, x_2)+\frac 23(1, 0)$, and  $S_3(x_1, x_2)=\frac 13 (x_1, x_2)+ \frac 23(\frac 12, \frac {\sqrt{3}}{2})$ for all $(x_1, x_2) \in \D R^2$.  The limit set of the iterated function system $\{S_i \}_{i=1}^3 $ is called a stretched Sierpi\'{n}ski triangle. Let $P=\frac 13\mathop{\sum}_{j=1}^3 P\circ S_j^{-1}$. Then, $P$ is a unique Borel probability measure on $\D R^2$ with support the stretched Sierpi\'{n}ski triangle generated by $S_1, S_2, S_3$.  For this probability measure $P$, \c C\"omez and Roychowdhury determined the optimal sets of $n$-means and the $n$th quantization errors for all $n\geq 2$. Further, they showed that although the quantization dimension exists, the quantization coefficient for the probability measure $P$ does not exist (see \cite{CR2}).
\par


In this paper, we have considered a set of three contractive similarity mappings $S_1, S_2, S_3$ on $\D R^2$, such that
$S_1(x_1, x_2)=\frac 14 (x_1, x_2)$, $S_2(x_1, x_2)=\frac 14 (x_1, x_2)+\frac 34(1, 0)$, and  $S_3(x_1, x_2)=\frac 12 (x_1, x_2)+ \frac 12(\frac 12, \frac {\sqrt{3}}{2})$ for all $(x_1, x_2) \in \D R^2$.  In this case, we call the limit set, denoted by $S$,  as a \tit{nonuniform stretched Sierpi\'{n}ski triangle} generated by the contractive mappings $S_1, S_2, S_3$. The term `nonuniform' is used to mean that the basic triangles at each level in the construction of the stretched Sierpi\'{n}ski triangle are not of equal shape. Let $P=\frac 15P\circ S_1^{-1}+\frac 15 P\circ S_2^{-1}+\frac 3 5 P\circ S_3^{-1}$. Then, $P$ is a unique Borel probability measure on $\D R^2$ with support the nonuniform stretched Sierpi\'{n}ski triangle generated by $S_1, S_2, S_3$. 

\subsection{Delineation} 
For the probability measure $P$  with support the nonuniform stretched Sierpi\'{n}ski triangle generated by $S_1, S_2, S_3$, in this paper, in Theorem~\ref{Th1}, we state and prove an induction formula to determine the optimal sets of $n$-means for all $n\geq 2$. Once the optimal sets are known, the corresponding quantization errors can easily be obtained. We also give some figures to illustrate the locations of the elements in the optimal sets (see Figure~\ref{Fig1}, Figure~\ref{Fig2} and Figure~\ref{Fig3}). In addition, using the induction formula, we obtain some results and observations about the optimal sets of $n$-means which are given in Section~4; a tree diagram of the optimal sets of $n$-means for a certain range of $n$ is also given (see Figure~\ref{Fig4}).

\subsection{Significance of the work} Over the time, quantization dimensions and quantization coefficients for different fractal probability measures were investigated by many researchers. On the other hand, finding the optimal sets of $n$-means and the $n$th quantization errors are much more difficult. The work in this paper is an endeavor in this direction. By using the methodology given in this paper, or with an overhaul, one can calculate the optimal sets of $n$-means and the $n$th quantization errors for more general fractal probability measures.

\begin{figure} 
\begin{tikzpicture}[line cap=round,line join=round,>=triangle 45,x=1.0cm,y=1.0cm]
\clip(-0.787435429074699,-0.985782091795184) rectangle (8.972433259875483,8.36691883633011);
\draw [line width=0.5 pt] (0.,0.)-- (8.,0.);
\draw [line width=0.5 pt] (0.,0.)-- (4.,6.928203230275509);
\draw [line width=0.5 pt] (8.,0.)-- (4.,6.928203230275509);
\draw [line width=0.2 pt,dash pattern=on 2pt off 2pt] (4.,0.)-- (4.,6.928203230275509);
\draw [line width=0.5 pt] (2.,0.)-- (1.,1.7320508075688772);
\draw [line width=0.5 pt] (6.,0.)-- (7.,1.7320508075688772);
\draw [line width=0.5 pt] (2.,3.4641016151377544)-- (6.,3.4641016151377544);
\draw [line width=0.5 pt] (0.5,0.866025)-- (1.5,0.866025);
\draw [line width=0.5 pt] (0.,0.)-- (0.5,0.);
\draw [line width=0.5 pt] (0.25,0.433013)-- (0.5,0.);
\draw [line width=0.5 pt] (1.5,0.)-- (1.75,0.433013);
\draw [line width=0.5 pt] (7.5,0.)-- (7.75,0.433013);
\draw [line width=0.5 pt] (6.5,0.)-- (6.25,0.433013);
\draw [line width=0.5 pt] (6.5,0.866025)-- (7.5,0.866025);
\draw [line width=0.5 pt] (3.,5.19615)-- (5.,5.19615);
\draw [line width=0.5 pt] (3.,3.4641)-- (2.5,4.33013);
\draw [line width=0.5 pt] (5.,3.4641)-- (5.5,4.33013);
\draw [line width=0.5 pt] (3.5,5.19615)-- (3.25,5.62917);
\draw [line width=0.5 pt] (4.5,5.19615)-- (4.75,5.62917);
\draw [line width=0.5 pt] (3.5,6.06218)-- (4.5,6.06218);
\draw (3.7193157594453146,6.004417138366035) node[anchor=north west] {$\triangle_{33} $};
\draw (2.157303081167313,4.0020195834781775) node[anchor=north west] {$\triangle_{31}$};
\draw (5.205899413499102,3.9220195834781775) node[anchor=north west] {$\triangle_{32}$};
\draw (3.7362684422614577,4.448120556984892) node[anchor=north west] {$\triangle_3$};
\draw (0.7291957685215965,0.6417071542495328) node[anchor=north west] {$\triangle_1$};
\draw (6.7709594089609615,0.6417071542495328) node[anchor=north west] {$\triangle_2$};
\draw (0.6137667442973821,1.3883791035320332) node[anchor=north west] {$\triangle_{13}$};
\draw (1.7662525944139552,0.6417071542495328) node[anchor=north west] {$\triangle_{12}$};
\draw (-0.5641481300434051,0.6417071542495328) node[anchor=north west] {$\triangle_{11}$};
\draw (6.660101360512533,1.3883791035320332) node[anchor=north west] {$\triangle_{23}$};
\draw (5.4221864861717455,0.6417071542495328) node[anchor=north west] {$\triangle_{21}$};
\draw (7.756492576261391,0.6417071542495328) node[anchor=north west] {$\triangle_{22}$};
\draw (3.7193157594453146,2.3643854426710336) node[anchor=north west] {$\triangle$};
\draw (2.9235018586112405,5.208795675836889) node[anchor=north west] {$\triangle_{331}$};
\draw (4.388561854073099, 5.208795675836889) node[anchor=north west] {$\triangle_{332}$};
\draw (3.7193157594453146,6.59417138366035) node[anchor=north west] {$\triangle_{333}$};
\draw (-0.7858642269402625,0.07331910807017522) node[anchor=north west] {$(0, 0)$};
\draw (7.955826966117891,0.14722447370246097) node[anchor=north west] {$(1,0)$};
\draw (3.956460880566387,7.570335182276322) node[anchor=north west] {$(\frac 12, \frac{\sqrt 3} {2})$};
\draw (0.744432354440882,3.887448607702391) node[anchor=north west] {$(\frac{1}{4},\frac{\sqrt{3}}{4})$};
\draw (5.888000387005818,3.887448607702391) node[anchor=north west] {$(\frac{3}{4},\frac{\sqrt{3}}{4})$};
\draw (1.7247289358220265,0.05874813229438953) node[anchor=north west] {$(\frac{1}{4},0)$};
\draw (5.602378924476674,0.05874813229438953) node[anchor=north west] {$(\frac{3}{4},0)$};
\draw (-0.2271954472272622,2.157198325275336) node[anchor=north west] {$(\frac{1}{8},\frac{\sqrt{3}}{8})$};
\draw (6.926580871490104,2.157198325275336) node[anchor=north west] {$(\frac{7}{8},\frac{\sqrt{3}}{8})$};
\draw (2.439211616369099,3.594682024052177) node[anchor=north west] {$(\frac{3}{8},\frac{\sqrt{3}}{4})$};
\draw (4.3020823430955295,3.594682024052177) node[anchor=north west] {$(\frac{5}{8},\frac{\sqrt{3}}{4})$};
\draw (0.9368140614812393,4.741265678105964) node[anchor=north west] {$(\frac{5}{16},\frac{5 \sqrt{3}}{16})$};
\draw (5.4333044534620174,4.741265678105964) node[anchor=north west] {$(\frac{11}{16},\frac{5 \sqrt{3}}{16})$};
\draw (1.5728204006202403,5.585748358653035) node[anchor=north west] {$(\frac{3}{8},\frac{3 \sqrt{3}}{8})$};
\draw (4.941801609561888,5.585748358653035) node[anchor=north west] {$(\frac{5}{8},\frac{3 \sqrt{3}}{8})$};
\end{tikzpicture}
\vspace{-0.2 in} 
\caption{Some basic triangles with their vertices that construct the stretched Sierpi\'{n}ski triangle.}\label{Fig0}
\end{figure}

\section{Basic definitions and lemmas}

In this section, we give the basic definitions and lemmas that will be instrumental in our analysis. By a \textit{string} or a \textit{word} $\go$ over an alphabet $I:=\{1, 2, 3\}$, we mean a finite sequence $\go:=\go_1\go_2\cdots \go_k$
of symbols from the alphabet, where $k\geq 1$, and $k$ is called the length of the word $\go$.  A word of length zero is called the \textit{empty word} and is denoted by $\emptyset$.  By $I^*$, we denote the set of all words
over the alphabet $I$ of some finite length $k$, including the empty word $\emptyset$. By $|\go|$, we denote the length of a word $\go \in I^*$. For any two words $\go:=\go_1\go_2\cdots \go_k$ and
$\tau:=\tau_1\tau_2\cdots \tau_\ell$ in $I^*$, by
$\go\tau:=\go_1\cdots \go_k\tau_1\cdots \tau_\ell$ we mean the word obtained from the concatenation of $\go$ and $\tau$.  As defined in the previous section, the mappings $S_i :\D R^2 \to \D R^2$ are the generating mappings of the nonuniform stretched Sierpi\'{n}ski triangle with similarity ratios $s_i$ for $1\leq i\leq 3$, respectively, and $P=\mathop{\sum}\limits_{i=1}^3 p_i P\circ S_i^{-1}$ is the probability distribution, where $s_1=s_2=\frac 14$, $s_3=\frac 12$, $p_1=p_2=\frac 15$ and $p_3=\frac 35$. In short, the `nonuniform stretched Sierpi\'{n}ski triangle'  in the sequel will be referred to as `stretched Sierpi\'{n}ski triangle'. 
For $\go=\go_1\go_2 \cdots\go_k \in I^k$, set
$S_\go:=S_{\go_1} \circ S_{\go_2}\circ  \cdots\circ S_{\go_k},  \ s_\go:=s_{\go_1}s_{\go_2}\cdots s_{\go_k} \te{ and } p_\go:=p_{\go_1}p_{\go_2}\cdots p_{\go_k}$.
Let $\tri$ be the equilateral triangle with vertices $(0, 0)$, $(1, 0)$ and $(\frac 12, \frac {\sqrt{3}} {2})$, and $\tri_\go=S_\go(\tri)$ for $\go=\go_1\go_2 \cdots\go_k \in I^k$. The sets $\set{\tri_\go : \go \in I^k}$ are just the $3^k$ triangles in the $k$th level in the construction of the stretched Sierpi\'{n}ski triangle. The triangles $\tri_{\go1}$, $\tri_{\go2}$ and $\tri_{\go3}$ into which $\tri_\go$ is split up at the $(k+1)$th level are called the \tit{basic triangles} of $\tri_\go$ (see Figure~\ref{Fig0}). The set $S:=\II_{k \in \D N} \UU_{\go \in I^k} \tri_\go$ is the stretched Sierpi\'{n}ski triangle and equals the support of the probability measure $P$. For $\go=\go_1 \go_2 \cdots \go_k \in I^k$, let us write
$c(\go):=\#\{i : \go_i=3, \, 1\leq i\leq k\}$. Then, we have
\[P(\tri_\go)=p_\go=\frac{3^{c(\go)}}{5^{|\go|}} \te{ and } s_\go=\frac{2^{c(\go)}}{4^{|\go|}}.\]

Let us now give the following lemma.
\begin{lemma} \label{lemma1} Let $f: \D R \to \D R^+$ be Borel measurable and $k\in \D N$. Then,
\[\int f \,dP=\sum_{\go \in I^k} p_\go\int f\circ S_\go \,dP.\]
\end{lemma}
\begin{proof}
We know $P=\mathop{\sum}\limits_{i=1}^3 p_i P\circ S_i^{-1}$, and so by induction $P=\mathop{\sum}\limits_{\go \in I^k}  p_\go P\circ S_\go^{-1}$, and thus the lemma is yielded.
\end{proof}

Let us now state the following lemma. The proof is routine (see \cite{CR2, R4}).
\begin{lemma} \label{lemma333} Let $E(X)$ and $V(X)$ denote the expected vector and the variance of the random variable $X$. Then, \[E(X)=(E(X_1), \, E(X_2))=(\frac{1}{2},\frac{\sqrt{3}}{4}) \te{ and } V:=V(X)=E\|X-(\frac 1 2, \frac {\sqrt {3}}{4})\|^2=\frac{27}{176}\]
with $V(X_1)=\frac 3 {44}$ and $V(X_2)=\frac{15}{176}$.
\end{lemma}

Let us now give the following note.
\begin{note} \label{note1} From Lemma~\ref{lemma333}, it follows that the optimal set of one-mean is the expected vector, and the corresponding quantization error is the variance $V$ of the random variable $X$. For words $\gb, \gg, \cdots, \gd$ in $I^\ast$, by $a(\gb, \gg, \cdots, \gd)$ we mean the conditional expected vector of the random variable $X$ given $\tri_\gb\uu \tri_\gg \uu\cdots \uu \tri_\gd,$ i.e.,
\begin{equation} \label{eq010}
a(\gb, \gg, \cdots, \gd)=E(X|X\in \tri_\gb \uu \tri_\gg \uu \cdots \uu \tri_\gd)=\frac{1}{P(\tri_\gb\uu \cdots \uu \tri_\gd)}\int_{\tri_\gb\uu \cdots \uu \tri_\gd} x dP.
\end{equation}
For $\go \in I^k$, $k\geq 1$, since $a(\go)=E(X : X \in J_\go)$, using Lemma~\ref{lemma1}, we have
\begin{align*}
&a(\go)=\frac{1}{P(\tri_\go)} \int_{\tri_\go} x \,dP(x)=\int_{\tri_\go} x\, dP\circ S_\go^{-1}(x)=\int S_\go(x)\, dP(x)=E(S_\go(X))=S_\go(\frac 12, \frac {\sqrt{3}}{4}).
\end{align*}  For any $(a, b) \in \D R^2$, $E\|X-(a, b)\|^2=V+\|(\frac 12, \frac {\sqrt{3}} 4)-(a, b)\|^2.$
In fact, for any $\go \in I^k$, $k\geq 1$, we have
$\int_{\tri_\go}\|x-(a, b)\|^2 dP= p_\go \int\|(x_1, x_2) -(a, b)\|^2 dP\circ S_\go^{-1},$
which implies
\begin{equation} \label{eq1}
\int_{\tri_\go}\|x-(a, b)\|^2 dP=p_\go \Big(s_\go^2V+\|a(\go)-(a, b)\|^2\Big).
\end{equation}
The expressions \eqref{eq010} and \eqref{eq1} are useful to obtain the optimal sets and the corresponding quantization errors with respect to the probability distribution $P$.  Notice that with respect to the median passing through the vertex $(\frac 12, \frac{\sqrt{3}}2)$, the stretched Sierpi\'{n}ski triangle has the maximum symmetry, i.e., with respect to the line $x_1=\frac 12$ the stretched Sierpi\'{n}ski triangle is geometrically symmetric. Also, observe that if the two basic rectangles of similar geometrical shape lie on opposite sides of the line $x_1=\frac 12$, and are equidistant from the line $x_1=\frac 12$, then they have the same probability (see Figure~\ref{Fig0}); hence, they are symmetric with respect to the probability distribution $P$ as well.
\end{note}
In the next section, we determine the optimal sets of $n$-means for all $n\geq 2$.

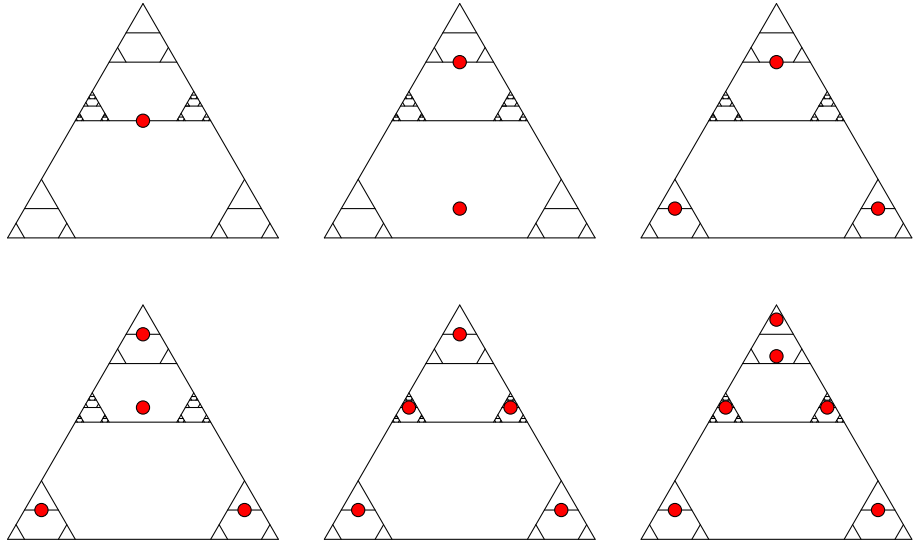
\begin{figure}
\begin{tikzpicture}[line cap=round,line join=round,>=triangle 45,x=0.3 cm,y=0.3 cm]
\clip(-0.717946985327632,-0.81755092664275) rectangle (12.8592934709513,11.586582306908248);
\draw (0.,0.)-- (6.,10.392304845413264);
\draw (6.,10.392304845413264)-- (12.,0.);
\draw (0.,0.)-- (12.,0.);
\draw (1.5,2.598076211353316)-- (3.,0.);
\draw (10.5,2.598076211353316)-- (9.,0.);
\draw (3.,5.196152422706632)-- (9.,5.196152422706632);
\draw (3.75,6.495190528383289)-- (4.5,5.196152422706632);
\draw (8.25,6.495190528383289)-- (7.5,5.196152422706632);
\draw (4.5,7.794228634059947)-- (7.5,7.794228634059947);
\draw (0.375,0.649519052838329)-- (0.75,0.);
\draw (2.625,0.649519052838329)-- (2.25,0.);
\draw (0.75,1.299038105676658)-- (2.25,1.299038105676658);
\draw (9.75,1.299038105676658)-- (11.25,1.299038105676658);
\draw (9.375,0.649519052838329)-- (9.75,0.);
\draw (11.625,0.649519052838329)-- (11.25,0.);
\draw (4.875,8.443747686898277)-- (5.25,7.794228634059947);
\draw (7.125,8.443747686898277)-- (6.75,7.794228634059947);
\draw (5.25,9.093266739736606)-- (6.75,9.093266739736606);
\draw (3.1875,5.520911949125796)-- (3.375,5.196152422706632);
\draw (4.3125,5.520911949125796)-- (4.125,5.196152422706632);
\draw (7.6875,5.520911949125796)-- (7.875,5.196152422706632);
\draw (8.8125,5.520911949125796)-- (8.625,5.196152422706632);
\draw (3.375,5.84567147554496)-- (4.125,5.84567147554496);
\draw (7.875,5.84567147554496)-- (8.625,5.84567147554496);
\draw (3.46875,6.008051238754542)-- (3.5625,5.84567147554496);
\draw (4.03125,6.008051238754542)-- (3.9375,5.84567147554496);
\draw (8.53125,6.008051238754542)-- (8.4375,5.84567147554496);
\draw (7.96875,6.008051238754542)-- (8.0625,5.84567147554496);
\draw (7.875,5.84567147554496)-- (8.0625,5.84567147554496);
\draw (3.5625,6.170431001964125)-- (3.9375,6.170431001964125);
\draw (8.0625,6.170431001964125)-- (8.4375,6.170431001964125);
\draw (3.046875,5.277342304311422)-- (3.09375,5.196152422706632);
\draw (3.328125,5.277342304311422)-- (3.28125,5.196152422706632);
\draw (3.09375,5.358532185916213)-- (3.28125,5.358532185916213);
\draw (3.65625,6.332810765173707)-- (3.84375,6.332810765173707);
\draw (3.609375,6.251620883568916)-- (3.65625,6.170431001964125);
\draw (3.890625,6.251620883568916)-- (3.84375,6.170431001964125);
\draw (3.1171875,5.399127126718609)-- (3.140625,5.358532185916213);
\draw (3.2578125,5.399127126718609)-- (3.234375,5.358532185916213);
\draw (3.140625,5.4397220675210045)-- (3.234375,5.4397220675210045);
\draw (4.171875,5.277342304311422)-- (4.21875,5.196152422706632);
\draw (4.453125,5.277342304311422)-- (4.40625,5.196152422706632);
\draw (4.21875,5.358532185916213)-- (4.40625,5.358532185916213);
\draw (4.2421875,5.399127126718609)-- (4.265625,5.358532185916213);
\draw (4.3828125,5.399127126718609)-- (4.359375,5.358532185916213);
\draw (4.265625,5.4397220675210045)-- (4.359375,5.4397220675210045);
\draw (7.546875,5.277342304311422)-- (7.59375,5.196152422706632);
\draw (7.828125,5.277342304311422)-- (7.78125,5.196152422706632);
\draw (8.953125,5.277342304311422)-- (8.90625,5.196152422706632);
\draw (8.671875,5.277342304311422)-- (8.71875,5.196152422706632);
\draw (8.71875,5.358532185916213)-- (8.90625,5.358532185916213);
\draw (7.59375,5.358532185916213)-- (7.78125,5.358532185916213);
\draw (8.15625,6.332810765173707)-- (8.34375,6.332810765173707);
\draw (8.109375,6.251620883568916)-- (8.15625,6.170431001964125);
\draw (8.390625,6.251620883568916)-- (8.34375,6.170431001964125);
\begin{scriptsize}
\draw [fill=ffqqqq] (6.,5.196152422706632) circle (2.5pt);
\end{scriptsize}
\end{tikzpicture}\
 \begin{tikzpicture}[line cap=round,line join=round,>=triangle 45,x=0.3 cm,y=0.3 cm]
\clip(-0.717946985327634,-0.8175509266427525) rectangle (12.859293470951293,11.58658230690825);
\draw (0.,0.)-- (6.,10.392304845413264);
\draw (6.,10.392304845413264)-- (12.,0.);
\draw (0.,0.)-- (12.,0.);
\draw (1.5,2.598076211353316)-- (3.,0.);
\draw (10.5,2.598076211353316)-- (9.,0.);
\draw (3.,5.196152422706632)-- (9.,5.196152422706632);
\draw (3.75,6.495190528383289)-- (4.5,5.196152422706632);
\draw (8.25,6.495190528383289)-- (7.5,5.196152422706632);
\draw (4.5,7.794228634059947)-- (7.5,7.794228634059947);
\draw (0.375,0.649519052838329)-- (0.75,0.);
\draw (2.625,0.649519052838329)-- (2.25,0.);
\draw (0.75,1.299038105676658)-- (2.25,1.299038105676658);
\draw (9.75,1.299038105676658)-- (11.25,1.299038105676658);
\draw (9.375,0.649519052838329)-- (9.75,0.);
\draw (11.625,0.649519052838329)-- (11.25,0.);
\draw (4.875,8.443747686898277)-- (5.25,7.794228634059947);
\draw (7.125,8.443747686898277)-- (6.75,7.794228634059947);
\draw (5.25,9.093266739736606)-- (6.75,9.093266739736606);
\draw (3.1875,5.520911949125796)-- (3.375,5.196152422706632);
\draw (4.3125,5.520911949125796)-- (4.125,5.196152422706632);
\draw (7.6875,5.520911949125796)-- (7.875,5.196152422706632);
\draw (8.8125,5.520911949125796)-- (8.625,5.196152422706632);
\draw (3.375,5.84567147554496)-- (4.125,5.84567147554496);
\draw (7.875,5.84567147554496)-- (8.625,5.84567147554496);
\draw (3.46875,6.008051238754542)-- (3.5625,5.84567147554496);
\draw (4.03125,6.008051238754542)-- (3.9375,5.84567147554496);
\draw (8.53125,6.008051238754542)-- (8.4375,5.84567147554496);
\draw (7.96875,6.008051238754542)-- (8.0625,5.84567147554496);
\draw (7.875,5.84567147554496)-- (8.0625,5.84567147554496);
\draw (3.5625,6.170431001964125)-- (3.9375,6.170431001964125);
\draw (8.0625,6.170431001964125)-- (8.4375,6.170431001964125);
\draw (3.046875,5.277342304311422)-- (3.09375,5.196152422706632);
\draw (3.328125,5.277342304311422)-- (3.28125,5.196152422706632);
\draw (3.09375,5.358532185916213)-- (3.28125,5.358532185916213);
\draw (3.65625,6.332810765173707)-- (3.84375,6.332810765173707);
\draw (3.609375,6.251620883568916)-- (3.65625,6.170431001964125);
\draw (3.890625,6.251620883568916)-- (3.84375,6.170431001964125);
\draw (3.1171875,5.399127126718609)-- (3.140625,5.358532185916213);
\draw (3.2578125,5.399127126718609)-- (3.234375,5.358532185916213);
\draw (3.140625,5.4397220675210045)-- (3.234375,5.4397220675210045);
\draw (4.171875,5.277342304311422)-- (4.21875,5.196152422706632);
\draw (4.453125,5.277342304311422)-- (4.40625,5.196152422706632);
\draw (4.21875,5.358532185916213)-- (4.40625,5.358532185916213);
\draw (4.2421875,5.399127126718609)-- (4.265625,5.358532185916213);
\draw (4.3828125,5.399127126718609)-- (4.359375,5.358532185916213);
\draw (4.265625,5.4397220675210045)-- (4.359375,5.4397220675210045);
\draw (7.546875,5.277342304311422)-- (7.59375,5.196152422706632);
\draw (7.828125,5.277342304311422)-- (7.78125,5.196152422706632);
\draw (8.953125,5.277342304311422)-- (8.90625,5.196152422706632);
\draw (8.671875,5.277342304311422)-- (8.71875,5.196152422706632);
\draw (8.71875,5.358532185916213)-- (8.90625,5.358532185916213);
\draw (7.59375,5.358532185916213)-- (7.78125,5.358532185916213);
\draw (8.15625,6.332810765173707)-- (8.34375,6.332810765173707);
\draw (8.109375,6.251620883568916)-- (8.15625,6.170431001964125);
\draw (8.390625,6.251620883568916)-- (8.34375,6.170431001964125);
\begin{scriptsize}
\draw [fill=ffqqqq] (6.,1.299038105676658) circle (2.5pt);
\draw [fill=ffqqqq] (6.,7.794228634059947) circle (2.5pt);
\end{scriptsize}
\end{tikzpicture}\
\begin{tikzpicture}[line cap=round,line join=round,>=triangle 45,x=0.3 cm,y=0.3 cm]
\clip(-0.717946985327634,-0.8175509266427525) rectangle (12.859293470951293,11.58658230690825);
\draw (0.,0.)-- (6.,10.392304845413264);
\draw (6.,10.392304845413264)-- (12.,0.);
\draw (0.,0.)-- (12.,0.);
\draw (1.5,2.598076211353316)-- (3.,0.);
\draw (10.5,2.598076211353316)-- (9.,0.);
\draw (3.,5.196152422706632)-- (9.,5.196152422706632);
\draw (3.75,6.495190528383289)-- (4.5,5.196152422706632);
\draw (8.25,6.495190528383289)-- (7.5,5.196152422706632);
\draw (4.5,7.794228634059947)-- (7.5,7.794228634059947);
\draw (0.375,0.649519052838329)-- (0.75,0.);
\draw (2.625,0.649519052838329)-- (2.25,0.);
\draw (0.75,1.299038105676658)-- (2.25,1.299038105676658);
\draw (9.75,1.299038105676658)-- (11.25,1.299038105676658);
\draw (9.375,0.649519052838329)-- (9.75,0.);
\draw (11.625,0.649519052838329)-- (11.25,0.);
\draw (4.875,8.443747686898277)-- (5.25,7.794228634059947);
\draw (7.125,8.443747686898277)-- (6.75,7.794228634059947);
\draw (5.25,9.093266739736606)-- (6.75,9.093266739736606);
\draw (3.1875,5.520911949125796)-- (3.375,5.196152422706632);
\draw (4.3125,5.520911949125796)-- (4.125,5.196152422706632);
\draw (7.6875,5.520911949125796)-- (7.875,5.196152422706632);
\draw (8.8125,5.520911949125796)-- (8.625,5.196152422706632);
\draw (3.375,5.84567147554496)-- (4.125,5.84567147554496);
\draw (7.875,5.84567147554496)-- (8.625,5.84567147554496);
\draw (3.46875,6.008051238754542)-- (3.5625,5.84567147554496);
\draw (4.03125,6.008051238754542)-- (3.9375,5.84567147554496);
\draw (8.53125,6.008051238754542)-- (8.4375,5.84567147554496);
\draw (7.96875,6.008051238754542)-- (8.0625,5.84567147554496);
\draw (7.875,5.84567147554496)-- (8.0625,5.84567147554496);
\draw (3.5625,6.170431001964125)-- (3.9375,6.170431001964125);
\draw (8.0625,6.170431001964125)-- (8.4375,6.170431001964125);
\draw (3.046875,5.277342304311422)-- (3.09375,5.196152422706632);
\draw (3.328125,5.277342304311422)-- (3.28125,5.196152422706632);
\draw (3.09375,5.358532185916213)-- (3.28125,5.358532185916213);
\draw (3.65625,6.332810765173707)-- (3.84375,6.332810765173707);
\draw (3.609375,6.251620883568916)-- (3.65625,6.170431001964125);
\draw (3.890625,6.251620883568916)-- (3.84375,6.170431001964125);
\draw (3.1171875,5.399127126718609)-- (3.140625,5.358532185916213);
\draw (3.2578125,5.399127126718609)-- (3.234375,5.358532185916213);
\draw (3.140625,5.4397220675210045)-- (3.234375,5.4397220675210045);
\draw (4.171875,5.277342304311422)-- (4.21875,5.196152422706632);
\draw (4.453125,5.277342304311422)-- (4.40625,5.196152422706632);
\draw (4.21875,5.358532185916213)-- (4.40625,5.358532185916213);
\draw (4.2421875,5.399127126718609)-- (4.265625,5.358532185916213);
\draw (4.3828125,5.399127126718609)-- (4.359375,5.358532185916213);
\draw (4.265625,5.4397220675210045)-- (4.359375,5.4397220675210045);
\draw (7.546875,5.277342304311422)-- (7.59375,5.196152422706632);
\draw (7.828125,5.277342304311422)-- (7.78125,5.196152422706632);
\draw (8.953125,5.277342304311422)-- (8.90625,5.196152422706632);
\draw (8.671875,5.277342304311422)-- (8.71875,5.196152422706632);
\draw (8.71875,5.358532185916213)-- (8.90625,5.358532185916213);
\draw (7.59375,5.358532185916213)-- (7.78125,5.358532185916213);
\draw (8.15625,6.332810765173707)-- (8.34375,6.332810765173707);
\draw (8.109375,6.251620883568916)-- (8.15625,6.170431001964125);
\draw (8.390625,6.251620883568916)-- (8.34375,6.170431001964125);
\begin{scriptsize}
\draw [fill=ffqqqq] (6., 7.79423) circle (2.5pt);
\draw [fill=ffqqqq] (1.5, 1.29904) circle (2.5pt);
\draw [fill=ffqqqq] (10.5, 1.29904) circle (2.5pt);
\end{scriptsize}
\end{tikzpicture}

\begin{tikzpicture}[line cap=round,line join=round,>=triangle 45,x=0.3 cm,y=0.3 cm]
\clip(-0.717946985327632,-0.81755092664275) rectangle (12.8592934709513,11.586582306908248);
\draw (0.,0.)-- (6.,10.392304845413264);
\draw (6.,10.392304845413264)-- (12.,0.);
\draw (0.,0.)-- (12.,0.);
\draw (1.5,2.598076211353316)-- (3.,0.);
\draw (10.5,2.598076211353316)-- (9.,0.);
\draw (3.,5.196152422706632)-- (9.,5.196152422706632);
\draw (3.75,6.495190528383289)-- (4.5,5.196152422706632);
\draw (8.25,6.495190528383289)-- (7.5,5.196152422706632);
\draw (4.5,7.794228634059947)-- (7.5,7.794228634059947);
\draw (0.375,0.649519052838329)-- (0.75,0.);
\draw (2.625,0.649519052838329)-- (2.25,0.);
\draw (0.75,1.299038105676658)-- (2.25,1.299038105676658);
\draw (9.75,1.299038105676658)-- (11.25,1.299038105676658);
\draw (9.375,0.649519052838329)-- (9.75,0.);
\draw (11.625,0.649519052838329)-- (11.25,0.);
\draw (4.875,8.443747686898277)-- (5.25,7.794228634059947);
\draw (7.125,8.443747686898277)-- (6.75,7.794228634059947);
\draw (5.25,9.093266739736606)-- (6.75,9.093266739736606);
\draw (3.1875,5.520911949125796)-- (3.375,5.196152422706632);
\draw (4.3125,5.520911949125796)-- (4.125,5.196152422706632);
\draw (7.6875,5.520911949125796)-- (7.875,5.196152422706632);
\draw (8.8125,5.520911949125796)-- (8.625,5.196152422706632);
\draw (3.375,5.84567147554496)-- (4.125,5.84567147554496);
\draw (7.875,5.84567147554496)-- (8.625,5.84567147554496);
\draw (3.46875,6.008051238754542)-- (3.5625,5.84567147554496);
\draw (4.03125,6.008051238754542)-- (3.9375,5.84567147554496);
\draw (8.53125,6.008051238754542)-- (8.4375,5.84567147554496);
\draw (7.96875,6.008051238754542)-- (8.0625,5.84567147554496);
\draw (7.875,5.84567147554496)-- (8.0625,5.84567147554496);
\draw (3.5625,6.170431001964125)-- (3.9375,6.170431001964125);
\draw (8.0625,6.170431001964125)-- (8.4375,6.170431001964125);
\draw (3.046875,5.277342304311422)-- (3.09375,5.196152422706632);
\draw (3.328125,5.277342304311422)-- (3.28125,5.196152422706632);
\draw (3.09375,5.358532185916213)-- (3.28125,5.358532185916213);
\draw (3.65625,6.332810765173707)-- (3.84375,6.332810765173707);
\draw (3.609375,6.251620883568916)-- (3.65625,6.170431001964125);
\draw (3.890625,6.251620883568916)-- (3.84375,6.170431001964125);
\draw (3.1171875,5.399127126718609)-- (3.140625,5.358532185916213);
\draw (3.2578125,5.399127126718609)-- (3.234375,5.358532185916213);
\draw (3.140625,5.4397220675210045)-- (3.234375,5.4397220675210045);
\draw (4.171875,5.277342304311422)-- (4.21875,5.196152422706632);
\draw (4.453125,5.277342304311422)-- (4.40625,5.196152422706632);
\draw (4.21875,5.358532185916213)-- (4.40625,5.358532185916213);
\draw (4.2421875,5.399127126718609)-- (4.265625,5.358532185916213);
\draw (4.3828125,5.399127126718609)-- (4.359375,5.358532185916213);
\draw (4.265625,5.4397220675210045)-- (4.359375,5.4397220675210045);
\draw (7.546875,5.277342304311422)-- (7.59375,5.196152422706632);
\draw (7.828125,5.277342304311422)-- (7.78125,5.196152422706632);
\draw (8.953125,5.277342304311422)-- (8.90625,5.196152422706632);
\draw (8.671875,5.277342304311422)-- (8.71875,5.196152422706632);
\draw (8.71875,5.358532185916213)-- (8.90625,5.358532185916213);
\draw (7.59375,5.358532185916213)-- (7.78125,5.358532185916213);
\draw (8.15625,6.332810765173707)-- (8.34375,6.332810765173707);
\draw (8.109375,6.251620883568916)-- (8.15625,6.170431001964125);
\draw (8.390625,6.251620883568916)-- (8.34375,6.170431001964125);
\begin{scriptsize}
\draw [fill=ffqqqq] (6., 9.09327) circle (2.5pt);
\draw [fill=ffqqqq] (6., 5.84567) circle (2.5pt);
\draw [fill=ffqqqq] (1.5, 1.29904) circle (2.5pt);
\draw [fill=ffqqqq] (10.5, 1.29904) circle (2.5pt);
\end{scriptsize}
\end{tikzpicture}\
 \begin{tikzpicture}[line cap=round,line join=round,>=triangle 45,x=0.3 cm,y=0.3 cm]
\clip(-0.717946985327634,-0.8175509266427525) rectangle (12.859293470951293,11.58658230690825);
\draw (0.,0.)-- (6.,10.392304845413264);
\draw (6.,10.392304845413264)-- (12.,0.);
\draw (0.,0.)-- (12.,0.);
\draw (1.5,2.598076211353316)-- (3.,0.);
\draw (10.5,2.598076211353316)-- (9.,0.);
\draw (3.,5.196152422706632)-- (9.,5.196152422706632);
\draw (3.75,6.495190528383289)-- (4.5,5.196152422706632);
\draw (8.25,6.495190528383289)-- (7.5,5.196152422706632);
\draw (4.5,7.794228634059947)-- (7.5,7.794228634059947);
\draw (0.375,0.649519052838329)-- (0.75,0.);
\draw (2.625,0.649519052838329)-- (2.25,0.);
\draw (0.75,1.299038105676658)-- (2.25,1.299038105676658);
\draw (9.75,1.299038105676658)-- (11.25,1.299038105676658);
\draw (9.375,0.649519052838329)-- (9.75,0.);
\draw (11.625,0.649519052838329)-- (11.25,0.);
\draw (4.875,8.443747686898277)-- (5.25,7.794228634059947);
\draw (7.125,8.443747686898277)-- (6.75,7.794228634059947);
\draw (5.25,9.093266739736606)-- (6.75,9.093266739736606);
\draw (3.1875,5.520911949125796)-- (3.375,5.196152422706632);
\draw (4.3125,5.520911949125796)-- (4.125,5.196152422706632);
\draw (7.6875,5.520911949125796)-- (7.875,5.196152422706632);
\draw (8.8125,5.520911949125796)-- (8.625,5.196152422706632);
\draw (3.375,5.84567147554496)-- (4.125,5.84567147554496);
\draw (7.875,5.84567147554496)-- (8.625,5.84567147554496);
\draw (3.46875,6.008051238754542)-- (3.5625,5.84567147554496);
\draw (4.03125,6.008051238754542)-- (3.9375,5.84567147554496);
\draw (8.53125,6.008051238754542)-- (8.4375,5.84567147554496);
\draw (7.96875,6.008051238754542)-- (8.0625,5.84567147554496);
\draw (7.875,5.84567147554496)-- (8.0625,5.84567147554496);
\draw (3.5625,6.170431001964125)-- (3.9375,6.170431001964125);
\draw (8.0625,6.170431001964125)-- (8.4375,6.170431001964125);
\draw (3.046875,5.277342304311422)-- (3.09375,5.196152422706632);
\draw (3.328125,5.277342304311422)-- (3.28125,5.196152422706632);
\draw (3.09375,5.358532185916213)-- (3.28125,5.358532185916213);
\draw (3.65625,6.332810765173707)-- (3.84375,6.332810765173707);
\draw (3.609375,6.251620883568916)-- (3.65625,6.170431001964125);
\draw (3.890625,6.251620883568916)-- (3.84375,6.170431001964125);
\draw (3.1171875,5.399127126718609)-- (3.140625,5.358532185916213);
\draw (3.2578125,5.399127126718609)-- (3.234375,5.358532185916213);
\draw (3.140625,5.4397220675210045)-- (3.234375,5.4397220675210045);
\draw (4.171875,5.277342304311422)-- (4.21875,5.196152422706632);
\draw (4.453125,5.277342304311422)-- (4.40625,5.196152422706632);
\draw (4.21875,5.358532185916213)-- (4.40625,5.358532185916213);
\draw (4.2421875,5.399127126718609)-- (4.265625,5.358532185916213);
\draw (4.3828125,5.399127126718609)-- (4.359375,5.358532185916213);
\draw (4.265625,5.4397220675210045)-- (4.359375,5.4397220675210045);
\draw (7.546875,5.277342304311422)-- (7.59375,5.196152422706632);
\draw (7.828125,5.277342304311422)-- (7.78125,5.196152422706632);
\draw (8.953125,5.277342304311422)-- (8.90625,5.196152422706632);
\draw (8.671875,5.277342304311422)-- (8.71875,5.196152422706632);
\draw (8.71875,5.358532185916213)-- (8.90625,5.358532185916213);
\draw (7.59375,5.358532185916213)-- (7.78125,5.358532185916213);
\draw (8.15625,6.332810765173707)-- (8.34375,6.332810765173707);
\draw (8.109375,6.251620883568916)-- (8.15625,6.170431001964125);
\draw (8.390625,6.251620883568916)-- (8.34375,6.170431001964125);
\begin{scriptsize}
\draw [fill=ffqqqq] (6., 9.09327) circle (2.5pt);
\draw [fill=ffqqqq] (3.75, 5.84567) circle (2.5pt);
\draw [fill=ffqqqq] (8.25, 5.84567) circle (2.5pt);
\draw [fill=ffqqqq] (1.5, 1.29904) circle (2.5pt);
\draw [fill=ffqqqq] (10.5, 1.29904) circle (2.5pt);
\end{scriptsize}
\end{tikzpicture}\
\begin{tikzpicture}[line cap=round,line join=round,>=triangle 45,x=0.3 cm,y=0.3 cm]
\clip(-0.717946985327634,-0.8175509266427525) rectangle (12.859293470951293,11.58658230690825);
\draw (0.,0.)-- (6.,10.392304845413264);
\draw (6.,10.392304845413264)-- (12.,0.);
\draw (0.,0.)-- (12.,0.);
\draw (1.5,2.598076211353316)-- (3.,0.);
\draw (10.5,2.598076211353316)-- (9.,0.);
\draw (3.,5.196152422706632)-- (9.,5.196152422706632);
\draw (3.75,6.495190528383289)-- (4.5,5.196152422706632);
\draw (8.25,6.495190528383289)-- (7.5,5.196152422706632);
\draw (4.5,7.794228634059947)-- (7.5,7.794228634059947);
\draw (0.375,0.649519052838329)-- (0.75,0.);
\draw (2.625,0.649519052838329)-- (2.25,0.);
\draw (0.75,1.299038105676658)-- (2.25,1.299038105676658);
\draw (9.75,1.299038105676658)-- (11.25,1.299038105676658);
\draw (9.375,0.649519052838329)-- (9.75,0.);
\draw (11.625,0.649519052838329)-- (11.25,0.);
\draw (4.875,8.443747686898277)-- (5.25,7.794228634059947);
\draw (7.125,8.443747686898277)-- (6.75,7.794228634059947);
\draw (5.25,9.093266739736606)-- (6.75,9.093266739736606);
\draw (3.1875,5.520911949125796)-- (3.375,5.196152422706632);
\draw (4.3125,5.520911949125796)-- (4.125,5.196152422706632);
\draw (7.6875,5.520911949125796)-- (7.875,5.196152422706632);
\draw (8.8125,5.520911949125796)-- (8.625,5.196152422706632);
\draw (3.375,5.84567147554496)-- (4.125,5.84567147554496);
\draw (7.875,5.84567147554496)-- (8.625,5.84567147554496);
\draw (3.46875,6.008051238754542)-- (3.5625,5.84567147554496);
\draw (4.03125,6.008051238754542)-- (3.9375,5.84567147554496);
\draw (8.53125,6.008051238754542)-- (8.4375,5.84567147554496);
\draw (7.96875,6.008051238754542)-- (8.0625,5.84567147554496);
\draw (7.875,5.84567147554496)-- (8.0625,5.84567147554496);
\draw (3.5625,6.170431001964125)-- (3.9375,6.170431001964125);
\draw (8.0625,6.170431001964125)-- (8.4375,6.170431001964125);
\draw (3.046875,5.277342304311422)-- (3.09375,5.196152422706632);
\draw (3.328125,5.277342304311422)-- (3.28125,5.196152422706632);
\draw (3.09375,5.358532185916213)-- (3.28125,5.358532185916213);
\draw (3.65625,6.332810765173707)-- (3.84375,6.332810765173707);
\draw (3.609375,6.251620883568916)-- (3.65625,6.170431001964125);
\draw (3.890625,6.251620883568916)-- (3.84375,6.170431001964125);
\draw (3.1171875,5.399127126718609)-- (3.140625,5.358532185916213);
\draw (3.2578125,5.399127126718609)-- (3.234375,5.358532185916213);
\draw (3.140625,5.4397220675210045)-- (3.234375,5.4397220675210045);
\draw (4.171875,5.277342304311422)-- (4.21875,5.196152422706632);
\draw (4.453125,5.277342304311422)-- (4.40625,5.196152422706632);
\draw (4.21875,5.358532185916213)-- (4.40625,5.358532185916213);
\draw (4.2421875,5.399127126718609)-- (4.265625,5.358532185916213);
\draw (4.3828125,5.399127126718609)-- (4.359375,5.358532185916213);
\draw (4.265625,5.4397220675210045)-- (4.359375,5.4397220675210045);
\draw (7.546875,5.277342304311422)-- (7.59375,5.196152422706632);
\draw (7.828125,5.277342304311422)-- (7.78125,5.196152422706632);
\draw (8.953125,5.277342304311422)-- (8.90625,5.196152422706632);
\draw (8.671875,5.277342304311422)-- (8.71875,5.196152422706632);
\draw (8.71875,5.358532185916213)-- (8.90625,5.358532185916213);
\draw (7.59375,5.358532185916213)-- (7.78125,5.358532185916213);
\draw (8.15625,6.332810765173707)-- (8.34375,6.332810765173707);
\draw (8.109375,6.251620883568916)-- (8.15625,6.170431001964125);
\draw (8.390625,6.251620883568916)-- (8.34375,6.170431001964125);
\begin{scriptsize}
\draw [fill=ffqqqq] (6., 9.74279) circle (2.5pt);
\draw [fill=ffqqqq] (6., 8.11899) circle (2.5pt);
\draw [fill=ffqqqq] (3.75, 5.84567) circle (2.5pt);
\draw [fill=ffqqqq] (8.25, 5.84567) circle (2.5pt);
\draw [fill=ffqqqq] (1.5, 1.29904) circle (2.5pt);
\draw [fill=ffqqqq] (10.5, 1.29904) circle (2.5pt);
\end{scriptsize}
\end{tikzpicture}
\caption{Optimal configuration of $n$ points for $1\leq n\leq 6$.} \label{Fig1}
\end{figure}

\begin{figure}
\begin{tikzpicture}[line cap=round,line join=round,>=triangle 45,x=0.3 cm,y=0.3 cm]
\clip(-0.717946985327632,-0.81755092664275) rectangle (12.8592934709513,11.586582306908248);
\draw (0.,0.)-- (6.,10.392304845413264);
\draw (6.,10.392304845413264)-- (12.,0.);
\draw (0.,0.)-- (12.,0.);
\draw (1.5,2.598076211353316)-- (3.,0.);
\draw (10.5,2.598076211353316)-- (9.,0.);
\draw (3.,5.196152422706632)-- (9.,5.196152422706632);
\draw (3.75,6.495190528383289)-- (4.5,5.196152422706632);
\draw (8.25,6.495190528383289)-- (7.5,5.196152422706632);
\draw (4.5,7.794228634059947)-- (7.5,7.794228634059947);
\draw (0.375,0.649519052838329)-- (0.75,0.);
\draw (2.625,0.649519052838329)-- (2.25,0.);
\draw (0.75,1.299038105676658)-- (2.25,1.299038105676658);
\draw (9.75,1.299038105676658)-- (11.25,1.299038105676658);
\draw (9.375,0.649519052838329)-- (9.75,0.);
\draw (11.625,0.649519052838329)-- (11.25,0.);
\draw (4.875,8.443747686898277)-- (5.25,7.794228634059947);
\draw (7.125,8.443747686898277)-- (6.75,7.794228634059947);
\draw (5.25,9.093266739736606)-- (6.75,9.093266739736606);
\draw (3.1875,5.520911949125796)-- (3.375,5.196152422706632);
\draw (4.3125,5.520911949125796)-- (4.125,5.196152422706632);
\draw (7.6875,5.520911949125796)-- (7.875,5.196152422706632);
\draw (8.8125,5.520911949125796)-- (8.625,5.196152422706632);
\draw (3.375,5.84567147554496)-- (4.125,5.84567147554496);
\draw (7.875,5.84567147554496)-- (8.625,5.84567147554496);
\draw (3.46875,6.008051238754542)-- (3.5625,5.84567147554496);
\draw (4.03125,6.008051238754542)-- (3.9375,5.84567147554496);
\draw (8.53125,6.008051238754542)-- (8.4375,5.84567147554496);
\draw (7.96875,6.008051238754542)-- (8.0625,5.84567147554496);
\draw (7.875,5.84567147554496)-- (8.0625,5.84567147554496);
\draw (3.5625,6.170431001964125)-- (3.9375,6.170431001964125);
\draw (8.0625,6.170431001964125)-- (8.4375,6.170431001964125);
\draw (3.046875,5.277342304311422)-- (3.09375,5.196152422706632);
\draw (3.328125,5.277342304311422)-- (3.28125,5.196152422706632);
\draw (3.09375,5.358532185916213)-- (3.28125,5.358532185916213);
\draw (3.65625,6.332810765173707)-- (3.84375,6.332810765173707);
\draw (3.609375,6.251620883568916)-- (3.65625,6.170431001964125);
\draw (3.890625,6.251620883568916)-- (3.84375,6.170431001964125);
\draw (3.1171875,5.399127126718609)-- (3.140625,5.358532185916213);
\draw (3.2578125,5.399127126718609)-- (3.234375,5.358532185916213);
\draw (3.140625,5.4397220675210045)-- (3.234375,5.4397220675210045);
\draw (4.171875,5.277342304311422)-- (4.21875,5.196152422706632);
\draw (4.453125,5.277342304311422)-- (4.40625,5.196152422706632);
\draw (4.21875,5.358532185916213)-- (4.40625,5.358532185916213);
\draw (4.2421875,5.399127126718609)-- (4.265625,5.358532185916213);
\draw (4.3828125,5.399127126718609)-- (4.359375,5.358532185916213);
\draw (4.265625,5.4397220675210045)-- (4.359375,5.4397220675210045);
\draw (7.546875,5.277342304311422)-- (7.59375,5.196152422706632);
\draw (7.828125,5.277342304311422)-- (7.78125,5.196152422706632);
\draw (8.953125,5.277342304311422)-- (8.90625,5.196152422706632);
\draw (8.671875,5.277342304311422)-- (8.71875,5.196152422706632);
\draw (8.71875,5.358532185916213)-- (8.90625,5.358532185916213);
\draw (7.59375,5.358532185916213)-- (7.78125,5.358532185916213);
\draw (8.15625,6.332810765173707)-- (8.34375,6.332810765173707);
\draw (8.109375,6.251620883568916)-- (8.15625,6.170431001964125);
\draw (8.390625,6.251620883568916)-- (8.34375,6.170431001964125);
\begin{scriptsize}
\draw [fill=ffqqqq] (6., 9.74279) circle (2.5pt);
\draw [fill=ffqqqq] (6., 8.11899) circle (2.5pt);
\draw [fill=ffqqqq] (3.75, 5.84567) circle (2.5pt);
\draw [fill=ffqqqq] (8.25, 5.84567) circle (2.5pt);
\draw [fill=ffqqqq] (1.5, 1.94856) circle (2.5pt);
\draw [fill=ffqqqq] (1.5, 0.32476) circle (2.5pt);
\draw [fill=ffqqqq] (10.5, 1.29904) circle (2.5pt);
\end{scriptsize}
\end{tikzpicture}\
 \begin{tikzpicture}[line cap=round,line join=round,>=triangle 45,x=0.3 cm,y=0.3 cm]
\clip(-0.717946985327634,-0.8175509266427525) rectangle (12.859293470951293,11.58658230690825);
\draw (0.,0.)-- (6.,10.392304845413264);
\draw (6.,10.392304845413264)-- (12.,0.);
\draw (0.,0.)-- (12.,0.);
\draw (1.5,2.598076211353316)-- (3.,0.);
\draw (10.5,2.598076211353316)-- (9.,0.);
\draw (3.,5.196152422706632)-- (9.,5.196152422706632);
\draw (3.75,6.495190528383289)-- (4.5,5.196152422706632);
\draw (8.25,6.495190528383289)-- (7.5,5.196152422706632);
\draw (4.5,7.794228634059947)-- (7.5,7.794228634059947);
\draw (0.375,0.649519052838329)-- (0.75,0.);
\draw (2.625,0.649519052838329)-- (2.25,0.);
\draw (0.75,1.299038105676658)-- (2.25,1.299038105676658);
\draw (9.75,1.299038105676658)-- (11.25,1.299038105676658);
\draw (9.375,0.649519052838329)-- (9.75,0.);
\draw (11.625,0.649519052838329)-- (11.25,0.);
\draw (4.875,8.443747686898277)-- (5.25,7.794228634059947);
\draw (7.125,8.443747686898277)-- (6.75,7.794228634059947);
\draw (5.25,9.093266739736606)-- (6.75,9.093266739736606);
\draw (3.1875,5.520911949125796)-- (3.375,5.196152422706632);
\draw (4.3125,5.520911949125796)-- (4.125,5.196152422706632);
\draw (7.6875,5.520911949125796)-- (7.875,5.196152422706632);
\draw (8.8125,5.520911949125796)-- (8.625,5.196152422706632);
\draw (3.375,5.84567147554496)-- (4.125,5.84567147554496);
\draw (7.875,5.84567147554496)-- (8.625,5.84567147554496);
\draw (3.46875,6.008051238754542)-- (3.5625,5.84567147554496);
\draw (4.03125,6.008051238754542)-- (3.9375,5.84567147554496);
\draw (8.53125,6.008051238754542)-- (8.4375,5.84567147554496);
\draw (7.96875,6.008051238754542)-- (8.0625,5.84567147554496);
\draw (7.875,5.84567147554496)-- (8.0625,5.84567147554496);
\draw (3.5625,6.170431001964125)-- (3.9375,6.170431001964125);
\draw (8.0625,6.170431001964125)-- (8.4375,6.170431001964125);
\draw (3.046875,5.277342304311422)-- (3.09375,5.196152422706632);
\draw (3.328125,5.277342304311422)-- (3.28125,5.196152422706632);
\draw (3.09375,5.358532185916213)-- (3.28125,5.358532185916213);
\draw (3.65625,6.332810765173707)-- (3.84375,6.332810765173707);
\draw (3.609375,6.251620883568916)-- (3.65625,6.170431001964125);
\draw (3.890625,6.251620883568916)-- (3.84375,6.170431001964125);
\draw (3.1171875,5.399127126718609)-- (3.140625,5.358532185916213);
\draw (3.2578125,5.399127126718609)-- (3.234375,5.358532185916213);
\draw (3.140625,5.4397220675210045)-- (3.234375,5.4397220675210045);
\draw (4.171875,5.277342304311422)-- (4.21875,5.196152422706632);
\draw (4.453125,5.277342304311422)-- (4.40625,5.196152422706632);
\draw (4.21875,5.358532185916213)-- (4.40625,5.358532185916213);
\draw (4.2421875,5.399127126718609)-- (4.265625,5.358532185916213);
\draw (4.3828125,5.399127126718609)-- (4.359375,5.358532185916213);
\draw (4.265625,5.4397220675210045)-- (4.359375,5.4397220675210045);
\draw (7.546875,5.277342304311422)-- (7.59375,5.196152422706632);
\draw (7.828125,5.277342304311422)-- (7.78125,5.196152422706632);
\draw (8.953125,5.277342304311422)-- (8.90625,5.196152422706632);
\draw (8.671875,5.277342304311422)-- (8.71875,5.196152422706632);
\draw (8.71875,5.358532185916213)-- (8.90625,5.358532185916213);
\draw (7.59375,5.358532185916213)-- (7.78125,5.358532185916213);
\draw (8.15625,6.332810765173707)-- (8.34375,6.332810765173707);
\draw (8.109375,6.251620883568916)-- (8.15625,6.170431001964125);
\draw (8.390625,6.251620883568916)-- (8.34375,6.170431001964125);
\begin{scriptsize}
\draw [fill=ffqqqq] (6., 9.74279) circle (2.5pt);
\draw [fill=ffqqqq] (6., 8.11899) circle (2.5pt);
\draw [fill=ffqqqq] (3.75, 5.84567) circle (2.5pt);
\draw [fill=ffqqqq] (8.25, 5.84567) circle (2.5pt);
\draw [fill=ffqqqq] (1.5, 1.29904) circle (2.5pt);
\draw [fill=ffqqqq] (10.5, 1.94856) circle (2.5pt);
\draw [fill=ffqqqq] (10.5, 0.32476) circle (2.5pt);
\end{scriptsize}
\end{tikzpicture}
\caption{Optimal configuration of $n$ points for $n=7$.} \label{Fig2}
\end{figure}
\begin{figure} 
 \begin{tikzpicture}[line cap=round,line join=round,>=triangle 45,x=0.3 cm,y=0.3 cm]
\clip(-0.717946985327634,-0.8175509266427525) rectangle (12.859293470951293,11.58658230690825);
\draw (0.,0.)-- (6.,10.392304845413264);
\draw (6.,10.392304845413264)-- (12.,0.);
\draw (0.,0.)-- (12.,0.);
\draw (1.5,2.598076211353316)-- (3.,0.);
\draw (10.5,2.598076211353316)-- (9.,0.);
\draw (3.,5.196152422706632)-- (9.,5.196152422706632);
\draw (3.75,6.495190528383289)-- (4.5,5.196152422706632);
\draw (8.25,6.495190528383289)-- (7.5,5.196152422706632);
\draw (4.5,7.794228634059947)-- (7.5,7.794228634059947);
\draw (0.375,0.649519052838329)-- (0.75,0.);
\draw (2.625,0.649519052838329)-- (2.25,0.);
\draw (0.75,1.299038105676658)-- (2.25,1.299038105676658);
\draw (9.75,1.299038105676658)-- (11.25,1.299038105676658);
\draw (9.375,0.649519052838329)-- (9.75,0.);
\draw (11.625,0.649519052838329)-- (11.25,0.);
\draw (4.875,8.443747686898277)-- (5.25,7.794228634059947);
\draw (7.125,8.443747686898277)-- (6.75,7.794228634059947);
\draw (5.25,9.093266739736606)-- (6.75,9.093266739736606);
\draw (3.1875,5.520911949125796)-- (3.375,5.196152422706632);
\draw (4.3125,5.520911949125796)-- (4.125,5.196152422706632);
\draw (7.6875,5.520911949125796)-- (7.875,5.196152422706632);
\draw (8.8125,5.520911949125796)-- (8.625,5.196152422706632);
\draw (3.375,5.84567147554496)-- (4.125,5.84567147554496);
\draw (7.875,5.84567147554496)-- (8.625,5.84567147554496);
\draw (3.46875,6.008051238754542)-- (3.5625,5.84567147554496);
\draw (4.03125,6.008051238754542)-- (3.9375,5.84567147554496);
\draw (8.53125,6.008051238754542)-- (8.4375,5.84567147554496);
\draw (7.96875,6.008051238754542)-- (8.0625,5.84567147554496);
\draw (7.875,5.84567147554496)-- (8.0625,5.84567147554496);
\draw (3.5625,6.170431001964125)-- (3.9375,6.170431001964125);
\draw (8.0625,6.170431001964125)-- (8.4375,6.170431001964125);
\draw (3.046875,5.277342304311422)-- (3.09375,5.196152422706632);
\draw (3.328125,5.277342304311422)-- (3.28125,5.196152422706632);
\draw (3.09375,5.358532185916213)-- (3.28125,5.358532185916213);
\draw (3.65625,6.332810765173707)-- (3.84375,6.332810765173707);
\draw (3.609375,6.251620883568916)-- (3.65625,6.170431001964125);
\draw (3.890625,6.251620883568916)-- (3.84375,6.170431001964125);
\draw (3.1171875,5.399127126718609)-- (3.140625,5.358532185916213);
\draw (3.2578125,5.399127126718609)-- (3.234375,5.358532185916213);
\draw (3.140625,5.4397220675210045)-- (3.234375,5.4397220675210045);
\draw (4.171875,5.277342304311422)-- (4.21875,5.196152422706632);
\draw (4.453125,5.277342304311422)-- (4.40625,5.196152422706632);
\draw (4.21875,5.358532185916213)-- (4.40625,5.358532185916213);
\draw (4.2421875,5.399127126718609)-- (4.265625,5.358532185916213);
\draw (4.3828125,5.399127126718609)-- (4.359375,5.358532185916213);
\draw (4.265625,5.4397220675210045)-- (4.359375,5.4397220675210045);
\draw (7.546875,5.277342304311422)-- (7.59375,5.196152422706632);
\draw (7.828125,5.277342304311422)-- (7.78125,5.196152422706632);
\draw (8.953125,5.277342304311422)-- (8.90625,5.196152422706632);
\draw (8.671875,5.277342304311422)-- (8.71875,5.196152422706632);
\draw (8.71875,5.358532185916213)-- (8.90625,5.358532185916213);
\draw (7.59375,5.358532185916213)-- (7.78125,5.358532185916213);
\draw (8.15625,6.332810765173707)-- (8.34375,6.332810765173707);
\draw (8.109375,6.251620883568916)-- (8.15625,6.170431001964125);
\draw (8.390625,6.251620883568916)-- (8.34375,6.170431001964125);
\begin{scriptsize}
\draw [fill=ffqqqq] (6., 9.74279) circle (2.5pt);
\draw [fill=ffqqqq] (6., 8.11899) circle (2.5pt);
\draw [fill=ffqqqq] (3.75, 5.84567) circle (2.5pt);
\draw [fill=ffqqqq] (8.25, 5.84567) circle (2.5pt);
\draw [fill=ffqqqq] (1.5, 1.94856) circle (2.5pt);
\draw [fill=ffqqqq] (1.5, 0.32476) circle (2.5pt);
\draw [fill=ffqqqq] (10.5, 1.94856) circle (2.5pt);
\draw [fill=ffqqqq] (10.5, 0.32476) circle (2.5pt);
\end{scriptsize}
\end{tikzpicture}
\caption{Optimal configuration of $n$ points for $n=8$.}\label{Fig3}
\end{figure}
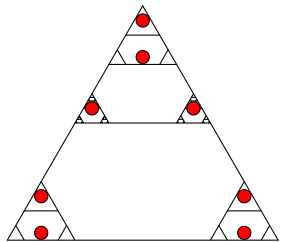
\section{Optimal sets of $n$-means for all $n\geq 2$}
 In this section, let us first prove the following proposition.
\begin{prop}\label{prop1} The set $\ga:=\set{a(1, 2), a(3)}$, where  $a(1, 2)=(\frac{1}{2},\frac{\sqrt{3}}{16})$ and $a(3)=(\frac{1}{2},\frac{3 \sqrt{3}}{8})$, is an optimal set of two-means with quantization error $V_2=\frac{117}{1408}\approx0.0830966$.
\end{prop}
\begin{proof} Let us consider the set of two elements $\gb$ given by $\gb:=\set{a(1, 2), a(3)}=\set{(\frac{1}{2},\frac{\sqrt{3}}{16}), (\frac{1}{2},\frac{3 \sqrt{3}}{8})}$. Then, $\tri_1\uu \tri_2\sci M(a(1,2)|\gb)$ and $\tri_3\sci M(a(3)|\gb)$, and so the distortion error due to the set $\gb$ is given by
\[\int\min_{b\in \gb}\|x-b\|^2 dP=\mathop{\int}\limits_{\tri_1\uu \tri_2}\|x-a(1,2)\|^2 dP+\mathop{\int}\limits_{\tri_3}\|x-a(3)\|^2 dP=\frac{117}{1408}\approx0.0830966.\]
Since $V_2$ is the quantization error for two-means, we have $V_2\leq 0.0830966$.
Notice that $x_1=\frac 12$ is the only line of maximum symmetry of the stretched Sierpi\'{n}ski triangle. Let $\ga:=\set{(a, b), (c, d)}$ be an optimal set of two-means. 

Then, one of the following two possible cases will happen.

\tit{Case~1. $(a, b)$ and $(c, d)$ are symmetrically located on both sides of the line $x_1=\frac 12$}, or \tit{Case~2. $(a, b)$ and $(c, d)$ lie on the vertical line $x_1=\frac 12$.}

Suppose that Case~1 happens. Then, recall that the elements in an optimal set are the conditional expectations in their own Voronoi regions. Write $A:=\tri_1\uu \tri_{31}\uu \tri_{331}\uu \tri_{3331}\uu \cdots$ and $B:=\tri_2\uu \tri_{32}\uu \tri_{332}\uu \tri_{3332}\uu \cdots$. See Figure~\ref{Fig0}, to understand how the basic triangles have been chosen. Then, 
\begin{align*}
&(a, b)=E(X : X \in A)=(\frac{2}{7},0.433013),  \te{ and }  (c, d)=E(X : X \in B)=(\frac{5}{7},0.433013),
\end{align*}
which yields the distortion error as
\begin{equation*} \label{eq100} \int \min_{q\in \set{(a, b), \, (c, d)}}\|x-q\|^2 dP=\mathop{\int}\limits_A \|x-(a, b)\|^2 dP+\mathop{\int}\limits_B \|x-(c, d)\|^2 dP=\frac{927}{8624}\approx0.107491>V_2,\end{equation*}
which leads to a contradiction. Hence,  Case~1 cannot happen, i.e., the only possible case is Case~2, i.e., we can assume that $(a, b)$ and $(c, d)$ lie on the vertical line $x_1=\frac 12$. 

In this case, we have $a=c=\frac 12$, i.e., $\ga=\{(\frac 12, b), (\frac 12, d)\}$. Without any loss of generality, we can assume that $b<d$. Since the optimal elements are the centroids of their own Voronoi regions, we have $\ga\sci \tri$. Moreover, by the properties of centroids, we have
\[(\frac 12, b) P(M((\frac 12, b)|\ga))+(\frac 12, d) P(M((p,d)|\ga))=(\frac 12, \frac {\sqrt{3}} 4),\]
which implies that $b  P(M((p, b )|\ga))+ dP(M((p, d)|\ga))=\frac {\sqrt{3}} 4$. Thus, it follows that the two optimal elements  $(\frac 12, b)$ and $(\frac 12, d)$ lie on the opposite sides of the point $(\frac 12, \frac {\sqrt{3}} 4)$, and so we have 
\begin{equation} \label{eqM1} 0<b \leq  \frac {\sqrt{3}} 4\leq d<\frac {\sqrt 3}{2}.
\end{equation} 
 If the Voronoi region of $(\frac 12, d)$ contains elements from $\tri_1\uu \tri_2$, by Figure~\ref{Fig0}, we must have $\frac 12(b+d)<\frac {\sqrt 3} {8}$, i.e.,  $b<\frac {\sqrt 3}{4}-d<0$, which contradicts \eqref{eqM1}. Hence, the Voronoi region of $(\frac 12, d)$ does not contain any element from $\tri_1\uu \tri_2$. Again, notice that 
 \[a(1, 2)=(\frac{1}{2},\frac{\sqrt{3}}{16}).\]
 Hence, $\frac{\sqrt{3}}{16}\leq b \leq  \frac {\sqrt{3}} 4\leq d<\frac {\sqrt 3}{2}$.  Notice that $b\leq \frac{\sqrt{3}}{4}$ implies $\frac 12(b+d)\leq \frac 12(\frac{\sqrt{3}}{4}+\frac {\sqrt 3}{2})=\frac{3\sqrt 3}{8}$, and so $\tri_{33}\sci M((\frac 12, d)|\ga)$.
 We now show that the Voronoi region of $(\frac 12, b)$ does not contain any element from $\tri_3$. For the sake of contradiction, assume that the Voronoi region of $(\frac 12, b)$ contains elements from $\tri_3$, in other words, we can say that 
\begin{equation*} \label{eqM423}  M((\frac 12, b )|\ga)\ii \tri_3\neq \es \te{ with } P(M((\frac 12, b )|\ga)\ii\tri_{3})>0, 
\end{equation*} 
equivalently, we can say that 
\begin{equation} \label{Megha567} \frac {\sqrt{3}} 4< \frac 12(b+d)\leq \frac{3\sqrt 3}{8}.
\end{equation} 

Suppose that $\frac {5\sqrt 3} {16} \leq \frac 12(b+d)\leq \frac{3\sqrt 3}{8}$. 
 Then,  
  \begin{align*}
& \int\min_{c\in\ga}\|x-c\|^2 dP\geq  \mathop{\int}\limits_{\tri_{33}} \|x-a(33)\|^2 dP+\mathop{\int}\limits_{\tri_{1}\uu \tri_{2}\uu \tri_{31}\uu\tri_{32}}  \|x-(\frac{1}{2}, a(1, 2, 31, 32))\|^2 dP\\
&=\frac{4239}{45056}\approx 0.0940829>V_2,
\end{align*}
 which leads to a contradiction. Hence, \eqref{Megha567} reduces to 
 \begin{equation} \label{Megha568} \frac {\sqrt{3}} 4<\frac 12(b+d)\leq \frac {5\sqrt 3} {16}.
\end{equation} 
Next, suppose that  $  \frac{37 \sqrt{3}}{128}<\frac 12(b+d)\leq \frac{5 \sqrt{3}}{16}$, where $S_{3131}(\frac{1}{2},\frac{\sqrt{3}}{2})=(\frac{37}{128},\frac{37 \sqrt{3}}{128})$ and $S_{31}(\frac 12,\frac {\sqrt 3}{2})=(\frac{5}{16},\frac{5 \sqrt{3}}{16})$. Then, we have 
  \begin{align*}
  \int\min_{c\in\ga}\|x-c\|^2 dP&\geq  \mathop{\int}\limits_{\tri_{33}} \|x-a(33)\|^2 dP\\
&\qquad +\mathop{\int}\limits_{\tri_{1}\uu \tri_{2}\uu \tri_{311}\uu \tri_{312}\uu\tri_{321}\uu\tri_{322} \uu\tri_{3131\uu\tri_{3132}\uu\tri_{3231}\uu\tri_{3232}}}  \|x-Av\|^2 dP\\
&=\frac{812073627}{9743360000}\approx 0.0833464>V_2,
\end{align*}
where $Av=a(1, 2, 311, 312, 321, 322, 3131, 3132, 3231, 3232)$, which leads to a contradiction. 
Hence, \eqref{Megha568} reduces to 
 \begin{equation*} \label{Megha569} \frac {\sqrt{3}} 4<\frac 12(b+d)\leq \frac{37 \sqrt{3}}{128}.
\end{equation*} 
Proceeding in this way, we can show that as long as $\frac {\sqrt{3}} 4<\frac 12(b+d)$, a contradiction arises. Thus, we can deduce that $\frac 12(b+d)\leq \frac {\sqrt{3}} 4$, i.e.,  we can assume that the Voronoi region of $(\frac 12, b)$ does not contain any element from $\tri_3$ yielding
$(\frac 12, b)=a(1, 2)=(\frac{1}{2},\frac{\sqrt{3}}{16})$ and $(\frac 12, d)=a(3)=(\frac{1}{2},\frac{3 \sqrt{3}}{8})$. Hence, the set $\ga=\set{a(1, 2), a(3)}$ is an optimal set of two-means with quantization error $V_2=\frac{117}{1408}\approx0.0830966$, which is the proposition.
\end{proof}

\begin{remark} The set $\ga$ in the above proposition is a unique optimal set of two-means.
\end{remark}

Let us now prove the following proposition.

\begin{prop}\label{prop2}
Let $\ga$ be an optimal set of three-means. Then, $\ga=\set{a(1), a(2), a(3)}$ and $V_3=\frac{189}{7040}\approx0.0268466$, where $a(1)=(\frac{1}{8},\frac{\sqrt{3}}{16})$, $a(2)=(\frac{7}{8},\frac{\sqrt{3}}{16})$, and $a(3)=(\frac{1}{2},\frac{3 \sqrt{3}}{8})$. Moreover, the Voronoi region of the element $\ga\ii \tri_i$ does not contain any element from $\tri_j$ for all $1\leq j\neq i\leq 3$.
\end{prop}
\begin{proof}
Let us consider the three-point set $\gb$ given by $\gb:=\set{a(1), a(2), a(3)}$. Then, the distortion error is given by
\begin{align*}
\int\min_{a\in \ga} \|x-a\|^2 dP=\mathop{\sum}_{i=1}^3 \mathop{\int}\limits_{\tri_i}\|x-a(i)\|^2 dP=\frac{189}{7040}\approx0.0268466.
\end{align*}
Since $V_3$ is the quantization error for three-means, we have $V_3\leq 0.0268466$. Let $\ga$ be an optimal set of three-means. As the optimal elements are the centroids of their own Voronoi regions, we have $\ga \sci \tri$. Write $\ga:=\set{(a_i, b_i) : 1\leq i\leq 3}$. Since $(\frac 12, \frac{\sqrt{3}} 4)$ is the centroid of the stretched Sierpi\'{n}ski triangle, we have
\begin{equation} \label{eq44} \sum_{i=1}^3 (a_i, b_i) P(M((a_i, b_i)|\ga))=(\frac 12, \frac{\sqrt{3}} 4).\end{equation}
Suppose $\ga$ does not contain any element from $\tri_3$. Then, $b_i<\frac{\sqrt{3}} 4$ for all $1\leq i\leq 3$ implying
\[\sum_{i=1}^3 b_i P(M((a_i, b_i)|\ga))<\frac{\sqrt{3}} 4 \sum_{i=1}^3 P(M((a_i, b_i)|\ga))=\frac{\sqrt{3}} 4,\] which contradicts \eqref{eq44}.
 So, we can assume that $\ga$ contains an element from $\tri_3$. Similarly, we can prove that $\ga$ contains an element from $\tri\setminus \tri_3$. We now prove the following claim. 
 
 \tit{Claim. Let $\ga$ be an optimal set of three-means as described before. Then, $\ga$ contains only one element from $\tri_3$.}
 
To prove the claim, we proceed as follows:  We have already seen that $\ga$ contains an element from $\tri_3$ and an element from $\tri\setminus \tri_3$.  
 For the sake of contradiction, assume that $\ga$ contains two elements, say $(a_1, b_1)$ and $(a_2, b_2)$ from $\tri_3$, and an element $(a_3, b_3)$ from $\tri\setminus \tri_3$. 
 
 Then, the following two cases can happen. 
 
\tit{Case~1. Both the elements $(a_1, b_1)$ and $(a_2, b_2)$ lie on the vertical line $x_1=\frac 12$.}

Then, the union of the Voronoi regions of the elements $(a_1, b_1)$ and $(a_2, b_2)$ does not contain any element from $\tri_1\uu \tri_2$, which yields the fact that
 \begin{equation} \label{M101}  \begin{aligned}
&\int\min_{c\in \ga}\|x-c\|^2 dP\geq\mathop{\int}\limits_{\tri_{1}\uu\tri_2}\min_{c\in \ga}\|x-c\|^2 dP\geq \mathop{\int}\limits_{\tri_{1}\uu\tri_2}\|x-a(1,2)\|^2 dP \\
&=\frac{423}{7040}\approx0.0600852>V_3,
\end{aligned}
\end{equation} 
which leads to a contradiction. 

\tit{Case~2.  $(a_1, b_1)$ and $(a_2, b_2)$ lie on both sides of the vertical line $x_1=\frac 12$.}

In this case, due to the maximum symmetry with respect to the line $x_1=\frac 12$, we can assume that $(a_1, b_1)$ and $(a_2, b_2)$ are symmetrically located on both sides of the line $x_1=\frac 12$. Then, obviously due to symmetry, $(a_3, b_3)$ will lie on the line $x_1=\frac 12$. 

In this case, the following two subcases can happen.

\tit{Subcase~(i). The union of the Voronoi regions of  $(a_1, b_1)$ and $(a_2, b_2)$ does not contain any element from $\tri_1\uu \tri_2$.} 

In this subcase, as shown in \eqref{M101}, a contradiction arises. 

\tit{Subcase~(ii). The union of the Voronoi regions of  $(a_1, b_1)$ and $(a_2, b_2)$  contains elements from $\tri_1\uu \tri_2$.} 

Notice that by Proposition~\ref{prop1}, we know that $\ga_2:=\set{a(1, 2), a(3)}$ is an optimal set of two-means. Then, it is not difficult to show that $S_3(\ga_2)$ is an optimal set of two-means with respect to the image measure $P\circ S_3^{-1}$. Take $\gg:=\set{(\frac 14, \frac{\sqrt 3} 4), (\frac 34, \frac{\sqrt 3} 4)}$ (to know how $\gg$ is chosen, see Figure~\ref{Fig0}). If the union of the Voronoi regions of  $(a_1, b_1)$ and $(a_2, b_2)$  contains elements from $\tri_1\uu \tri_2$, then due to symmetry, the Voronoi regions of the elements of $\gg$ must contain elements from $\tri_1\uu \tri_2$; in fact, from Figure~\ref{Fig0}, by drawing the Voronoi regions of the elements of $\gg$ and $(a_3, b_3)$, we see that the union of the Voronoi regions of the elements of $\gg$ must contain $\tri_{133}\uu\tri_{233}$.  Thus, in this subcase, using the symmetry, we have  
  \begin{align*}
&\int\min_{c\in \ga}\|x-c\|^2 dP\geq\mathop{\int}\limits_{\tri_3}\min_{c\in \set{(a_1, b_1), (a_2, b_2)}}\|x-c\|^2 dP+  \mathop{\int}\limits_{\tri_{133}\uu \tri_{233}}\min_{c\in \gg}\|x-c\|^2 dP   +\mathop{\int}\limits_{\tri_{12}\uu \tri_{21}}\|x-a(12, 21)\|^2 dP \\
&\geq\mathop{\int}\limits_{\tri_3}\min_{c\in S_3(\ga_2)}\|x-c\|^2 dP+2 \mathop{\int}\limits_{\tri_{133}}\|x-(\frac 14, \frac {\sqrt 3}{4})\|^2 dP  +\mathop{\int}\limits_{\tri_{12}\uu \tri_{21}}\|x-a(12, 21)\|^2 dP \\
&\geq \frac 3{20} V_2+2 \mathop{\int}\limits_{\tri_{133}}\|x-(\frac 14, \frac {\sqrt 3}{4})\|^2 dP  +\mathop{\int}\limits_{\tri_{12}\uu \tri_{21}}\|x-a(12, 21)\|^2 dP \\
&=\frac{83691}{2816000}\approx 0.0297198>V_3,
\end{align*}
which leads to a contradiction.  

Taking into account, Case~1 and Case~2, we can conclude that $\ga$ contains only one element from $\tri_3$. Thus, the claim is true. 

By the claim, we conclude that $\ga$ contains only one element from $\tri_3$ and two elements from $\tri\setminus \tri_3$. Due to the maximum symmetry of the stretched Sierpi\'{n}ski triangle with respect to the line $x_1=\frac 12$, we can assume that the element of $\ga\ii \tri_3$, say $(a_1, b_1)$, lies on the line $x_1=\frac 12$, and the two elements of $\ga\ii (\tri\setminus \tri_3)$, say $(a_2, b_2)$ and $(a_3, b_3)$, are symmetrically distributed over the triangle $\tri$ with respect to the line $x_1=\frac 12$. Let $(a_2, b_2)$ and $(a_3, b_3)$ lie to the left and to the right of the line $x_1=\frac 12$, respectively. Notice that $\tri_1\sci M((a_2, b_2)|\ga)$, $\tri_{2}\sci M((a_3, b_3)|\ga)$, and the Voronoi regions of $(a_2, b_2)$ and $(a_3, b_3)$ do not contain any element from $\tri_{3}$.
 Hence, the optimal set of three-means is $\set{a(1), a(2), a(3)}$ and the quantization error is $V_3=\frac{189}{7040}\approx0.0268466$. By finding the perpendicular bisectors of the line segments joining the elements in $\ga$, we see that the perpendicular bisector of the line segments joining the elements $\ga\ii \tri_i$ and $\ga\ii \tri_j$
does not intersect any of $\tri_i$ or $\tri_j$ for $1\leq i\neq  j\leq 3$. Thus, the Voronoi region of the element $\ga\ii\tri_i$ does not contain any element from $\tri_j$ for all $1\leq j\neq i\leq 3$. Hence, the proof of the proposition is complete.
\end{proof}

\begin{prop} \label{prop3}
Let $\ga_n$ be an optimal set of $n$-means for all $n\geq 3$. Then, the following properties are true: \\
 $(i)$ $\ga_n\ii \tri_i\neq \es$ for all $1\leq i\leq 3$,\\
  $(ii)$ $\ga_n$ does not contain any element from $\tri\setminus (\tri_1\uu \tri_2\uu\tri_3)$, and \\
  $(iii)$ the Voronoi region of any element in $\ga_n\ii \tri_i$ does not contain any element from $\tri_j$ for all $1\leq j\neq i\leq 3$.
\end{prop}
\begin{proof}
Let $\ga_n$ be an optimal set of $n$-means for $n\geq 3$. By Proposition~\ref{prop2}, we see that Proposition~\ref{prop3} is true for $n=3$. We now show that the proposition is true for $n\geq 4$. Consider the set of four elements $\gb:=\set{a(1), a(2), a(31, 32), a(33)}$. Since $V_n$ is the quantization error for $n$-means for $n\geq 4$, we have
\[V_n\leq V_4\leq \int\min_{b\in \gb}\|x-b\|^2 dP=\frac{459}{28160}\approx0.0162997.\]
If $\ga_n$ does not contain any element from $\tri_3$, then
\begin{align*}V_n& \geq \mathop{\int}\limits_{\tri_{33}} \min_{(a, b)\in \ga_n} \|(x_1, x_2)-(a, b)\|^2 dP\geq \|(\frac 12, \frac {3\sqrt 3}{8})-(\frac 12, \frac {\sqrt3}{4})\|^2P(\tri_{33})=\frac{27}{1600},
\end{align*}
implying $V_n\geq \frac{27}{1600}\approx0.016875>V_n$,
which leads to a contradiction. So, we can assume that $\ga_n\ii \tri_3\neq \es$. If $\ga_n\sci \tri_3$, i.e., if $\ga_n$ does not contain any element below the horizontal line $x_2=\frac {\sqrt 3}{4}$, then
\[V_n\geq 2 \int_{\tri_1}\min_{(a, b)\in\ga_n}\|(x_1, x_2)-(a, b)\|^2 dP\geq 2 \|S_1(\frac{1}{2},\frac{\sqrt{3}}{2})-S_3(0, 0)\|^2 P(\tri_1)=\frac{1}{40}\approx0.025>V_n.\]
which gives a contradiction. So, we can assume that $\ga_n$ contains elements below the horizontal line $x_2=\frac {\sqrt 3}{4}$. Suppose that $\ga_n$ contains only one element, say $(a_1, b_1)$ below the line $x_2=\frac {\sqrt 3}{4}$. Then, the following two cases can happen. 

\tit{Case~1.  The Voronoi region of any element in  $\ga_n\ii \tri_3$ does not contain any element from $\tri_1\uu \tri_2$.}

Then, 
\begin{align*}V_n& \geq \mathop{\int}\limits_{\tri_{1}\uu \tri_{2}} \|(x_1, x_2)-a(1, 2)\|^2 dP =\frac{423}{7040}\approx0.0600852>V_n,
\end{align*}
which is a contradiction. 

\tit{Case~2.  The Voronoi regions of the elements in  $\ga_n\ii \tri_3$ contain elements from $\tri_1\uu \tri_2$.}

Notice that the elements in $\tri_3$ closest to the elements in $\tri_1$ and $\tri_2$  are $S_3(0, 0)$ and $S_3(1, 0)$, respectively. 
First, suppose that the Voronoi regions of the elements in  $\ga_n\ii \tri_3$ contain both $\tri_{13}$ and $\tri_{23}$. Then, due to symmetry 
\begin{align*}V_n& \geq 2 \mathop{\int}\limits_{\tri_{13}} \|(x_1, x_2)-S_3(0, 0)\|^2 dP =\frac{771}{35200}\approx0.0219034>V_n,
\end{align*}
which leads to a contradiction. Similarly, we can show that for a positive integer $k$ if the Voronoi regions of the elements in  $\ga_n\ii \tri_3$ contain both $\tri_{13^k}$ and $\tri_{23^k}$, then a contradiction arises. Notice that $k$ cannot be large enough, otherwise, Case~2 will be reduced to Case~1, which will lead to another contradiction. Next, suppose that  for a positive integer $k$ the Voronoi regions of the elements in  $\ga_n\ii \tri_3$ contain only $\tri_{13^k}$ or $\tri_{23^k}$, then as optimal elements are the conditional expectation in their own Voronoi regions, we can see that a contradiction arises, i.e., the property of being the conditional expectation is violated. Due to too much technicality, we do not give the details of the proof in this context. Interested readers can verify it by drawing geometrical figures and drawing the Voronoi regions in GeoGebra.

Taking into account both Case~1 and Case~2, we can assume that $\ga_n$ contains at least two elements below the horizontal line $x_2=\frac {\sqrt 3}{4}$, and then due to maximum symmetry of $\tri_1$ and $\tri_2$ with respect to the line $x_1=\frac 12$ at least  one element will belong to $\tri_1$, and at least one element will belong to $\tri_2$. Thus, we see that $\ga_n\ii\tri_i \neq \es$ for all $1\leq i\leq 3$, which completes the proof of $(i)$. 

We now show that $\ga_n$ does not contain any element from $\tri\setminus (\tri_1\uu\tri_2\uu\tri_3)$. For the sake of contradiction, assume that $\ga_n$ contains at least one element, say $(a, b)$, from 
 $\tri\setminus (\tri_1\uu\tri_2\uu\tri_3)$. Recall that the Voronoi regions of the element $(a,b)$ must have positive probability. Moreover, as $\ga_n\ii \tri_i\neq \es$ for $1\leq i\leq 3$, it can be seen that if the Voronoi regions if $(a, b)$ contains elements from $\tri_i$, then the Voronoi region of $(a, b)$ does not contain elements from $\tri_j$ for $1\leq j\neq i\leq 3$, i.e.,  the distortion error can further be reduced by moving $(a, b)$ to $\tri_i$. Thus, if $\ga_n$ contains an element from $\tri\setminus (\tri_1\uu\tri_2\uu\tri_3)$, then a contradiction arises. Hence, $\ga_n$ does not contain any element from $\tri\setminus (\tri_1\uu\tri_2\uu\tri_3)$, i.e., $(ii)$ is true. 
 
 If $n=3$, then $(iii)$ is true by Proposition~\ref{prop2}. Let $n\geq 4$. Then, as $(iii)$ is true for $n=3$, due to the properties $(i)$ and $(ii)$, the property $(iii)$ is obviously true for all $n\geq 4$.
 Thus, the proof is yielded.   
\end{proof}

The following lemma is also true here.
\begin{lemma}\emph{(see \cite [Lemma~3.7]{CR2})} \label{lemma9} Let $P=\mathop{\sum}\limits_{\go \in I^k} p_\go P\circ S_\go^{-1}$ for some $k\geq 1$, and $\ga$ be an optimal set of $n$-means for $P$. Then, $\set{S_\go(a) : a \in \ga}$ is an optimal set of $n$-means for the image measure $P\circ S_\go^{-1}$. The converse is also true: If $\gb$ is an optimal set of $n$-means for the image measure $P\circ S_\go^{-1}$, then $\set{S_\go^{-1}(a) : a \in \gb}$ is an optimal set of $n$-means for $P$.
\end{lemma}

\begin{prop}\label{prop4}
Let $\ga_n$ be an optimal set of $n$-means for $n\geq 3$. Then, for $c\in \ga_n$ either $c=a(\go)$ or $c=a(\go1, \go2)$ for some $\go \in I^\ast$.
\end{prop}
\begin{proof}
Let $\ga_n$ be an optimal set of $n$-means for $n\geq 3$ and $c\in \ga_n$. Then, by Proposition~\ref{prop3}, we see that either $c\in \ga_n\ii \tri_i$ for some $1\leq i\leq 3$. Without any loss of generality, we can assume that  $c\in \ga_n\ii \tri_1$. If $\te{card}(\ga_n\ii \tri_1)=1$, then by Lemma~\ref{lemma9}, $S_1^{-1}(\ga_n\ii \tri_1)$ is an optimal set of one-mean yielding $c=S_1(\frac{1}{2}, \frac{\sqrt{3}}{4})=a(1)$. If $\te{card}(\ga_n\ii \tri_1)= 2$, then by Lemma~\ref{lemma9}, $S_1^{-1}(\ga_n\ii \tri_1)$ is an optimal set of two-means, i.e., $S_1^{-1}(\ga_n\ii \tri_1)=\set{a(1, 2), a(3)}$ yielding $c=a(11, 12)$ or $c=a(13)$. Similarly, if $\te{card}(\ga_n\ii \tri_1)= 3$, then $c=a(11), a(12)$, or $c=a(13)$. Let $\te{card}(\ga_n\ii \tri_1)\geq 4$. Then, as similarity mappings preserve the ratio of the distances of a element from any other two elements, using Proposition~\ref{prop3} again, we have $(\ga_n\ii \tri_1)\ii \tri_{1i}=\ga_n\ii \tri_{1i}\neq \es$ for $1\leq i\leq 3$, and $\ga_n\ii \tri_1=\uu_{i=1}^3(\ga_n\ii \tri_{1i})$. Without any loss of generality, assume that $c\in \ga_n\ii \tri_{11}$. If $\te{card}(\ga_n\ii \tri_{11})=1$, then $c=a(11)$. If $\te{card}(\ga_n\ii \tri_{11})=2$, then $c=a(111, 112)$ or $c=a(113)$. If $\te{card}(\ga_n\ii \tri_{11})=3$, then $c=a(111), a(112)$, or $c=a(113)$. If $\te{card}(\ga_n\ii \tri_{11})\geq 4$, then proceeding inductively in a similar way, we can find a word $\go\in I^\ast$ with $11\prec \go$, such that $c\in \ga_n\ii \tri_\go$. If $\te{card}(\ga_n\ii \tri_\go)=1$, then $c=a(\go)$. If $\te{card}(\ga_n\ii \tri_\go)=2$, then $c=a(\go1, \go2)$ or $a(\go3)$. If $\te{card}(\ga_n\ii \tri_\go)=3$, then $c=a(\go1), a(\go2)$, or $a(\go3)$. Thus, the proof of the proposition is yielded.
\end{proof}

\begin{note}\label{note2}
Let $\ga$ be an optimal set of $n$-means for some $n\geq 2$. Then, by Proposition~\ref{prop4}, for $a\in \ga$ we have $P$-almost surely, $M(a|\ga)=\tri_\go$ if $a=a(\go)$, and $M(a|\ga)=\tri_{\go1}\uu\tri_{\go2}$ if $a=a(\go1, \go2)$. For $\go \in I^\ast$, write
\begin{equation}\label{eq2}
E(\go):=\mathop{\int}\limits_{\tri_\go}\|x-a(\go)\|^2 dP \te{ and } E(\go1, \go2):=\mathop{\int}\limits_{\tri_{\go1}\uu\tri_{\go2}}\|x-a(\go1, \go2)\|^2 dP.
\end{equation}
\end{note}
Let us now give the following lemma.

\begin{lemma} \label{lemma10}
For any $\go \in I^\ast$, let $E(\go)$ and $E(\go1, \go2)$ be defined by \eqref{eq2}. Then, $E(\go1, \go2)=\frac{47}{18} E(\go3)=\frac{47}{120}E(\go)$, and $E(\go1)=E(\go2)=\frac 1{12} E(\go3)=\frac{1}{80}E(\go)$.
\end{lemma}

\begin{proof} By \eqref{eq1}, we have
\begin{align*}
&E(\go1, \go2)=\mathop{\int}\limits_{\tri_{\go1}\uu\tri_{\go2}}\|x-a(\go1, \go2)\|^2 dP=\mathop{\int}\limits_{\tri_{\go1}}\|x-a(\go1, \go2)\|^2 dP+\mathop{\int}\limits_{\tri_{\go2}}\|x-a(\go1, \go2)\|^2 dP\\
&=p_{\go1} (s_{\go1}^2V+\|a(\go1)-a(\go1, \go2)\|^2)+p_{\go2} (s_{\go2}^2V+\|a(\go2)-a(\go1, \go2)\|^2).
\end{align*}
Notice that
\begin{align*}
&\|a(\go1)-a(\go1, \go2)\|^2=\|S_{\go1}(\frac 12, \frac{\sqrt{3}}{4})-\frac 12\Big(S_{\go1}(\frac 12, \frac{\sqrt{3}}{4})+S_{\go2}(\frac 12, \frac{\sqrt{3}}{4})\Big)\|^2\\
=&\frac 14\|S_{\go1}(\frac 12, \frac{\sqrt{3}}{4})-S_{\go2}(\frac 12, \frac{\sqrt{3}}{4})\|^2=\frac 1 4 s_\go^2 \|S_1(\frac 12, \frac{\sqrt{3}}{4})-S_2(\frac 12, \frac{\sqrt{3}}{4})\|^2=\frac 9{64} s_\go^2,
\end{align*}
and similarly, $\|a(\go2)-a(\go1, \go2)\|^2=\frac 9{64} s_\go^2$. Thus, we obtain,
\begin{align*}
&E(\go1, \go2)=p_{\go1} (s_{\go1}^2V+\frac 9{64} s_\go^2)+p_{\go2} (s_{\go2}^2V+\frac 9{64} s_\go^2)=p_\go s_\go^2 V(p_1s_1^2+p_2s_2^2)+\frac 9{64}p_\go s_\go^2(p_1+p_2)
\end{align*}
yielding $E(\go1, \go2)=p_\go s_\go^2 V(p_1s_1^2+p_2 s_2^2+\frac 9{160} \frac 1 V)=p_\go s_\go^2 V\frac{47}{120} =\frac{47}{120}E(\go).$
Since $p_1=p_2$, $s_1=s_2$, we have $E(\go1)=p_{\go1}s_{\go1}^2V=\frac 1 {80} p_\go s_\go^2 V=E(\go2)$. Again, $E(\go3)=p_{\go3}s_{\go3}^2V=E(\go)p_3s_3^2=\frac 3{20}E(\go)$. Hence,
\[E(\go1, \go2)=\frac{47}{18} E(\go3)=\frac{47}{120} E(\go) \te{ and } E(\go1)=E(\go2)=\frac 1{12} E(\go3)=\frac{1}{80}E(\go),\]
which is the lemma.
\end{proof}

The following lemma plays an important role in proving the main theorem of the paper.
%
%
%
%
%
%
%
%

\begin{lemma} \label{lemma11}

Let $\go, \gt \in I^\ast$. Then,

$(i)$ $E(\go)> E(\gt)$ if and only if $E(\go1, \go2)+E(\go3)+E(\gt)< E(\go)+E(\gt1, \gt2)+E(\gt3)$;

$(ii)$ $E(\go)> \frac{96}{47} E(\gt1, \gt2)$ if and only if $E(\go1, \go2)+E(\go3)+E(\gt1, \gt2)< E(\go)+E(\gt1)+E(\gt2)$;

$(iii)$  $E(\go1, \go2)> \frac {47}{96}E(\gt)$ if and only if $E(\go1)+E(\go2)+E(\gt)< E(\go1, \go2)+E(\gt1, \gt2)+E(\gt3)$;

$(iv)$  $E(\go1, \go2)> E(\gt1, \gt2)$ if and only if $E(\go1)+E(\go2)+E(\gt1, \gt2)< E(\go1, \go2)+E(\gt1)+E(\gt2)$;

where for any $\go \in I^\ast$,  $E(\go)$ and $E(\go1, \go2)$ are defined by \eqref{eq2}.
\end{lemma}

\begin{proof} To prove $(i)$, using Lemma~\ref{lemma10}, we see that
\begin{align*}
LHS&=E(\go1, \go2)+E(\go3)+E(\gt)=(\frac {47}{120}+\frac{3}{20})E(\go)+E(\gt)=\frac{13}{24} E(\go)+E(\gt),\\
RHS&= E(\go)+E(\gt1, \gt2)+E(\gt3)=E(\go)+\frac{13}{24} E(\gt).
\end{align*}
Thus, $LHS< RHS$ if and only if $\frac{13}{24} E(\go)+E(\gt)< E(\go)+\frac{13}{24} E(\gt)$, which yields $E(\gt)<E(\go)$. Thus $(i)$ is proved. 
To prove $(ii)$ we proceed as follows:  
\begin{align*}
LHS&=E(\go1, \go2)+E(\go3)+E(\gt1, \gt2)=(\frac {47}{120}+\frac{6}{40})E(\go)+E(\gt 1, \gt 2)=\frac{13}{24} E(\go)+E(\gt 1, \gt 2),\\
RHS&= E(\go)+E(\gt1)+E(\gt2)=E(\go)+\frac{3}{47} E(\gt 1,\gt 2).
\end{align*}
Thus, $LHS< RHS$ if and only if $\frac{13}{24} E(\go)+E(\gt 1, \gt 2)<E(\go)+\frac{3}{47} E(\gt 1,\gt 2)$, which yields $E(\go)> \frac{96}{47} E(\gt1, \gt2)$. Thus $(ii)$ is proved. Proceeding in the similar way, $(iii)$ and $(iv)$ can be proved. Thus, the lemma is deduced.
\end{proof}

\begin{figure}\label{Figurepdf}
\vspace{-0.5 in}
\centerline{\includegraphics[width=7.5 in, height=8 in]{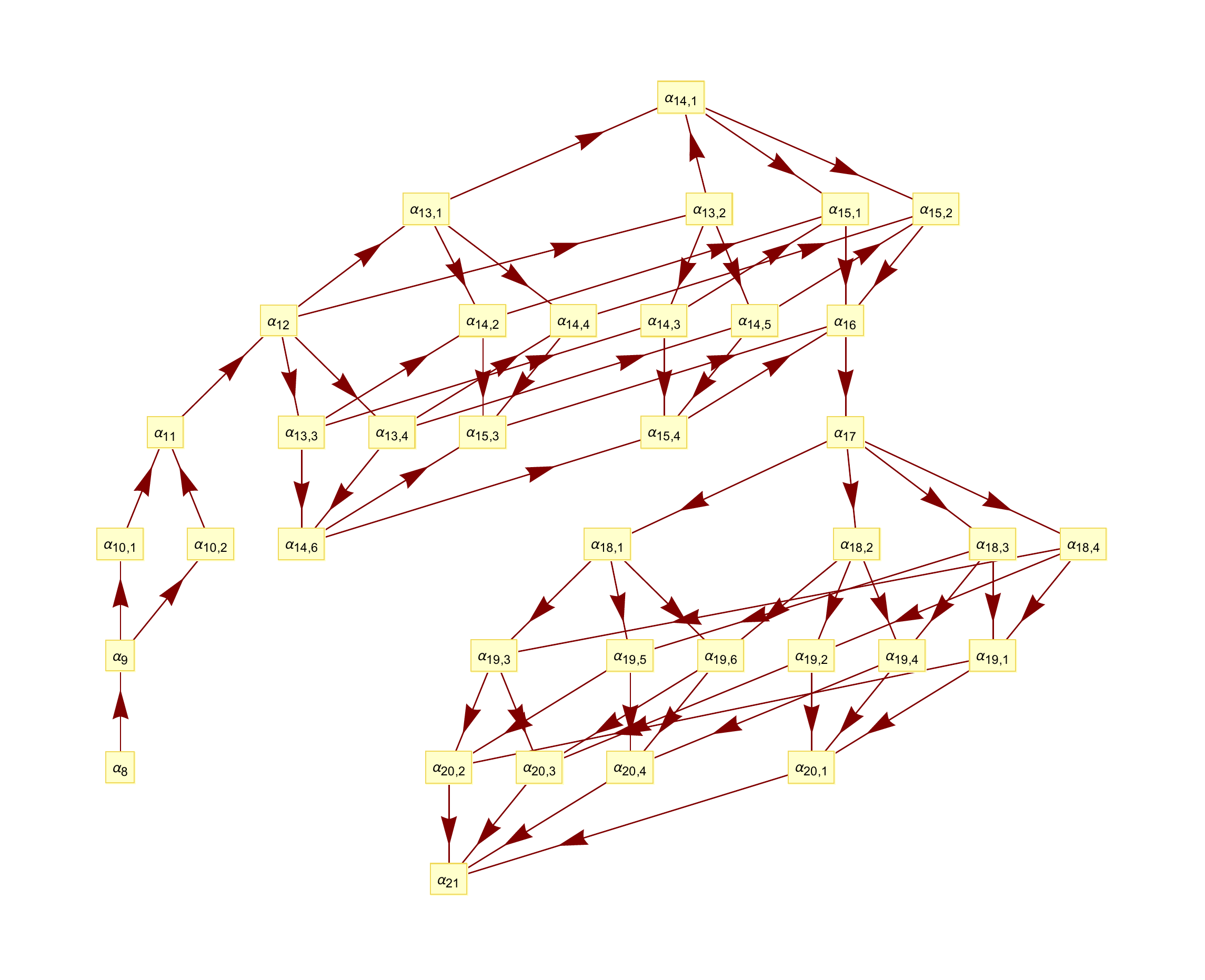}}
\vspace{-0.7 in}
\caption{Tree diagram of the optimal sets from $\ga_8$ to $\ga_{21}$.} \label{Fig4}
\end{figure}

In the following theorem, we give the induction formula to determine the optimal sets of $n$-means for any $n\geq 2$.

\begin{theorem} \label{Th1} For any $n\geq 2$, let $\ga_n:=\set{a(i) : 1\leq  i\leq n}$ be an optimal set of $n$-means, i.e., $\alpha_n \in\C C_n:= \mathcal{C}_n(P)$. For $\go \in I^\ast$, let $E(\go)$ and $E(\go1, \go2)$ be defined by \eqref{eq2}. Set
\[\tilde  E(a(i)):=\left\{\begin{array} {ll}
E(\go) \te{ if } a(i)=a(\go) \te{ for some }  \go \in I^\ast, \\
E(\go1, \go2) \te{ if } a(i)=a(\go1, \go2) \te{ for some }  \go \in I^\ast,
\end{array} \right.
\]
and $W(\ga_n):=\set{a(j)  : a(j) \in \ga_n \te{ and } \tilde E(a(j))\geq \tilde E(a(i)) \te{ for all } 1\leq i\leq n}$. Take any $a(j) \in W(\ga_n)$, and write
\[\ga_{n+1}(a(j)):=\left\{\begin{array}{ll}
(\ga_n\setminus \set{a(j)})\uu \set{a(\go1, \go2), a(\go3)} \te{ if } a(j)=a(\go), &\\
(\ga_n \setminus \set{a(j)})\uu \set{a(\go1), \, a(\go2)} \te{ if } a(j)=a(\go1, \go2).
\end{array}\right.
\]
Then $\ga_{n+1}(a(j))$ is an optimal set of $(n+1)$-means, and the number
of such sets is given by
\[\te{card}\Big(\UU_{\alpha_n \in \C{C}_n}\{\alpha_{n+1}(a(j)) : a(j) \in W(\ga_n)\}\Big).\]
\end{theorem}

\begin{proof}
By Proposition~\ref{prop1} and Proposition~\ref{prop2}, we know that the optimal sets of two- and three-means are $\{a(1,2), a(3)\}$ and
$\{a(1), a(2), a(3)\}$. Notice that by Lemma~\ref{lemma10}, we know $E(1, 2)=\frac {47}{18} E(3)>\frac {47}{96} E(3)$.
Hence, the theorem is true for $n=2$. For any $n\geq 2$, let us now assume that $\alpha_n$ is an optimal
set of $n$-means. Let $\ga_n:=\set{a(i) : 1\leq i\leq n}$. Let $\tilde  E(a(i))$ and $W(\ga_n)$ be defined as in the hypothesis. If $a(j) \not \in W(\ga_n)$, i.e., if  $a(j) \in \ga_n\setminus W(\ga_n)$, then by Lemma~\ref{lemma11}, the error
\[\sum_{a(i)\in (\ga_n\setminus \set{a(j)})}E(a(i))+E(\go1, \go2)+E(\go3) \te{ if } a(j)=a(\go),\]
or
\[\sum_{a(i)\in (\ga_n\setminus \set{a(j)})}E(a(i))+E(\go1)+E(\go2) \te{ if } a(j)=a(\go1, \go2),\]
obtained in this case is strictly greater than the corresponding error obtained in the case when $a(j)\in W(\ga_n)$. Hence, for any $a(j) \in W(\ga_n)$, the set $\ga_{n+1}(a(j))$, where
\[\ga_{n+1}(a(j)):=\left\{\begin{array}{ll}
(\ga_n\setminus \set{a(j)})\uu \set{a(\go1, \go2), a(\go3)} \te{ if } a(j)=a(\go), &\\
(\ga_n \setminus \set{a(j)})\uu \set{a(\go1), \, a(\go2)} \te{ if } a(j)=a(\go1, \go2).
\end{array}\right.
\] is an optimal set of $(n+1)$-means, and the number
of such sets is
\[\te{card}\Big(\UU_{\alpha_n \in \C{C}_n}\{\alpha_{n+1}(a(j)) : a(j) \in W(\ga_n)\}\Big).\]
Thus, the proof of the theorem is complete.
\end{proof}

\begin{remark}

Once an optimal set of $n$-means is known, by using \eqref{eq1}, the corresponding quantization error can easily be calculated.
\end{remark}

Using the induction formula given by Theorem~\ref{Th1}, we obtain some results and observations about the optimal sets of $n$-means, which are given in the following section.

\section{Some results and observations}
First, we explain some notations that we are going to use in this section. Recall that the optimal set of one-mean consists of the expected vector of the random vector $X$, and the corresponding quantization error is its variance. Let $\ga_n$ be an optimal set of $n$-means, i.e., $\ga_n \in \C C_n$, and then for any $a\in \ga_n$, we have  $a=a(\go)$, or $a=a(\go 1, \go 2)$ for some $\go \in I^k$, $k\geq 1$.
For any $n\geq 2$, if $\te{card}(\C C_n)=k$, we write
\[\C C_n=\left\{\begin{array}{ccc}
\set{\ga_{n, 1}, \ga_{n, 2}, \cdots, \ga_{n, k}} & \te{ if } k\geq 2,\\
\set{\ga_{n}} & \te{ if } k=1.
\end{array}\right.
\]
If $\te{card}(\C C_n)=k$ and  $\te{card}(\C C_{n+1})=m$, then either $1\leq k\leq m$, or $1\leq m\leq k$ (see Table~\ref{tab1}). Moreover, by Theorem~\ref{Th1}, an optimal set at stage $n$ can contribute multiple distinct optimal sets at stage $n+1$, and multiple distinct optimal sets at stage $n$ can contribute one common optimal set at stage $n+1$; for example from Table~\ref{tab1}, one can see that the number of $\ga_{12}=1$, the number of $\ga_{13}=4$, the number of $\ga_{14}=6$, the number of $\ga_{15}=4$, and the number of $\ga_{16}=1$.
\begin{table}[!ht]
\begin{center}
\begin{tabular}{ |c|c||c|c|| c|c||c|c|c||c|c||c|c}
 \hline
$n$ & $\te{card}(\C C_n) $ & $n$ & $\te{card}(\C C_n) $  & $n$ & $\te{card}(\C C_n)  $  & $n$ & $\te{card}(\C C_n)$   & $n$ & $\te{card}(\C C_n)$ & $n$ & $\te{card}(\C C_n)$\\
 \hline
5 & 1 & 18 &  4 &  31& 6 & 44 & 1& 57 & 495& 70 & 56 \\6 & 1 &  19 & 6 & 32 & 4 & 45& 8 & 58 & 792& 71 & 28\\7 & 2 & 20  & 4 & 33&  1& 46 & 28 & 59 & 924& 72 & 8\\8 & 1 &  21 &  1& 34 & 6& 47& 56 & 60&792 & 73 & 1\\9 & 1 & 22&  1 & 35 &  15& 48 & 70& 61 & 495 & 74 & 1\\10 & 2 & 23 & 6 & 36 & 20 & 49 & 56& 62 & 220& 75 & 12 \\11 & 1 & 24 &  15 & 37 & 15& 50  & 28& 63 & 66 & 76 & 66 \\12 & 1 & 25 & 20 & 38 & 6 & 51 & 8 & 64& 12& 77 & 220\\13 & 4 & 26 & 15 & 39& 1& 52 & 1& 65 & 1& 78 & 495 \\14 & 6 & 27& 6 & 40 & 1& 53 & 1& 66  & 8 & 79 & 792\\15 & 4 & 28&  1 & 41 &4 & 54 & 12& 67 & 28 & 80 & 924  \\16 & 1 &29 & 1 & 42 & 6 & 55 & 66& 68 & 56& 81 & 792 \\17 & 1 &30 & 4 & 43&  4 & 56 & 220 & 69 & 70 & 82 & 495\\
 \hline
\end{tabular}
 \end{center}
\caption{Number of $\ga_n$ in the range $5\leq n\leq 82$.}
    \label{tab1}
\end{table}
By $\ga_{n, i} \rightarrow \ga_{n+1, j}$, it is meant that the optimal set $\ga_{n+1, j}$ at stage $n+1$ is obtained from the optimal set $\ga_{n, i}$ at stage $n$, similar is the meaning for the notations $\ga_n\rightarrow \ga_{n+1, j}$, or $\ga_{n, i} \rightarrow \ga_{n+1}$, for example from Figure~\ref{Fig4}:
 \begin{align*} &\left\{\alpha _{12}\to \alpha _{13,1},\alpha _{12}\to \alpha _{13,2},\alpha _{12}\to \alpha _{13,3},\alpha _{12}\to \alpha _{13,4}\right\},\\
 &\{\left\{\alpha _{13,1}\to \alpha _{14,1},\alpha _{13,1}\to \alpha _{14,2},\alpha _{13,1}\to \alpha _{14,4}\right\},\left\{\alpha _{13,2}\to \alpha _{14,1},\alpha _{13,2}\to \alpha _{14,3},\alpha _{13,2}\to \alpha _{14,5}\right\},\\
 &\left\{\alpha _{13,3}\to \alpha _{14,2},\alpha _{13,3}\to \alpha _{14,3},\alpha _{13,3}\to \alpha _{14,6}\right\},\left\{\alpha _{13,4}\to \alpha _{14,4},\alpha _{13,4}\to \alpha _{14,5},\alpha _{13,4}\to \alpha _{14,6}\right\}\}.
 \end{align*}
Moreover, we see that
\begin{align*}
\ga_6&=\{a(1), a(2), a(31), a(32), a(333), a(331, 332) \}  \te{ with } V_6=\frac{3537}{563200}\approx0.00628018;\\
\ga_{7, 1}&=\{a(1), a(23), a(21, 22), a(31), a(32), a(333), a(331, 332) \};\\
\ga_{7, 2}&=\{a(13), a(11, 12),  a(2), a(31), a(32), a(333), a(331, 332) \}\\
& \qquad \qquad \te{ with }  V_7=\frac{1521}{281600}\approx0.00540128;\\
\ga_{8} &= \{a(13), a(11, 12),  a(23), a(21, 22), a(31), a(32), a(333), a(331, 332) \} \\
& \qquad \qquad  \te{ with } V_8=\frac{2547}{563200}\approx0.00452237; \\
\ga_9&= \{a(13), a(11, 12),  a(23), a(21, 22), a(31), a(32), a(333), a(331), a(332) \}  \\
& \qquad \qquad \te{ with } V_9=\frac{9171}{2816000}\approx0.00325675; \\
\ga_{10, 1}&=\{a(13), a(11, 12),  a(23), a(21), a(22), a(31), a(32), a(333), a(331), a(332)\}; \\
\ga_{10, 2}&=\{a(13), a(11), a(12),  a(23), a(21, 22), a(31), a(32), a(333), a(331), a(332) \}  \\
& \qquad \qquad \te{ with } V_{10} =\frac{ 7191}{2816000} \approx0.00255362;\\
\ga_{11} &=\{a(13), a(11), a(12),  a(23), a(21), a(22), a(31), a(32), a(333), a(331), a(332) \}  \\
&\qquad \qquad \te{ with }  V_{11}= \frac{5211}{2816000}\approx0.0018505;
\end{align*}
and so on.

\begin{remark}
By Theorem~\ref{Th1}, we see that to obtain an optimal set of $(n+1)$-means, one needs to know an optimal set of $n$-means. Unlike the probability distribution supported by the classical stretched Sierpi\'{n}ski triangle (see \cite{CR2}), for the probability distribution supported by the nonuniform stretched Sierpi\'{n}ski triangle considered in this paper, to obtain the optimal sets of $n$-means a closed formula is not known yet.
\end{remark}

\section*{Declaration}

\noindent
\textbf{Conflicts of interest.} We do not have any conflict of interest.\\
\\
\noindent
\textbf{Data availability:} No data were used to support this study.\\
\\
\noindent
\textbf{Code availability:} Not applicable\\
\\
\noindent
\textbf{Authors' contributions:} Each author contributed equally to this manuscript.









\bibliographystyle{abbrv}
\bibliography{References}

\medskip
Received xxxx 20xx; revised xxxx 20xx; early access xxxx 20xx.
\medskip

\end{document}